%% file: main.tex
\documentclass[11pt]{article}

\usepackage{fullpage}
\usepackage{amsmath,amsfonts,amsthm,mathrsfs,mathpazo,xspace,hyperref,graphicx}
\usepackage{endnotes}
\usepackage{color}
\usepackage{times}
\usepackage{amssymb,latexsym}
\usepackage{cite}

\usepackage[affil-it]{authblk}
\usepackage{hyperref}
\usepackage{braket}
\usepackage{esint}
\usepackage{bm}
\usepackage{graphicx}
\usepackage{caption}
\usepackage{subcaption}
\usepackage[all]{xy}
\usepackage{algorithm}

\usepackage[mathscr]{euscript}

\usepackage{tikz}
\usetikzlibrary{fit}

\usepackage{array}
\newcolumntype{C}{>{{}}c<{{}}}

\tikzset{sArrow/.style={->,>=stealth,thick}}
\tikzset{gArrow/.style={->,>=stealth,thick,gray}}
\tikzset{arrowLabel/.style={auto}}
\tikzset{blocResource/.style={draw,minimum width=4cm,minimum height=1.9cm}}
\tikzset{brnode/.style={minimum width=3.4cm,minimum height=.2cm}}
\tikzset{largeResource/.style={draw,minimum width=3.736cm,minimum height=3.25cm}}
\tikzset{medResource/.style={draw,minimum width=3.736cm,minimum height=1.75cm}}
\tikzset{thinResource/.style={draw,minimum width=1.618*2cm,minimum height=1cm}}
\tikzset{protocol/.style={draw,minimum width=1.545cm,minimum height=2.5cm}}
\tikzset{protocolLong/.style={draw,minimum height=1cm,minimum width=2.8cm}}
\tikzset{pnode/.style={minimum width=1cm,minimum height=.5cm}}
\tikzset{simulator/.style={draw,minimum width=2.8cm,minimum height=1.7cm}}

\DeclareMathAlphabet\mathbfcal{OMS}{cmsy}{b}{n}

\newtheorem{theorem}{Theorem}[section]

\newtheorem{lemma}[theorem]{Lemma}

\newtheorem{corollary}[theorem]{Corollary}

\theoremstyle{remark}

\theoremstyle{definition}
\newtheorem{definition}[theorem]{Definition}

\newcommand{\beq}{\begin{eqnarray}}
\newcommand{\eeq}{\end{eqnarray}}

\newcommand{\proj}[1]{\ket{#1}\!\bra{#1}}
\newcommand{\Tr}{\mbox{\rm Tr}}
\newcommand{\Id}{\ensuremath{\mathop{\rm Id}\nolimits}}
\newcommand{\Es}[1]{\textsc{E}_{#1}}

\newcommand{\reg}[1]{{\textsf{#1}}}

\newcommand{\C}{\ensuremath{\mathbb{C}}}
\newcommand{\N}{\ensuremath{\mathbb{N}}}

\newcommand{\R}{\ensuremath{\mathbb{R}}}
\newcommand{\Z}{\ensuremath{\mathbb{Z}}}

\newcommand{\mA}{\ensuremath{\mathcal{A}}}
\newcommand{\mE}{\ensuremath{\mathcal{E}}}
\newcommand{\mD}{\ensuremath{\mathcal{D}}}
\newcommand{\mF}{\ensuremath{\mathcal{F}}}
\newcommand{\mG}{\ensuremath{\mathcal{G}}}
\newcommand{\mK}{\ensuremath{\mathcal{K}}}
\newcommand{\mP}{\ensuremath{\mathcal{P}}}
\newcommand{\mS}{\ensuremath{\mathcal{S}}}

\newcommand{\mU}{\ensuremath{\mathcal{U}}}
\newcommand{\mX}{\ensuremath{\mathcal{X}}}
\newcommand{\mY}{\ensuremath{\mathcal{Y}}}

\newcommand{\Inv}{\ensuremath{\textsc{Inv}}}

\newcommand{\mH}{\mathcal{H}}

\newcommand{\opt}{\ensuremath{\textsc{opt}}}

\newcommand{\setft}[1]{\mathrm{#1}}
\newcommand{\Density}{\setft{D}}
\newcommand{\Pos}{\setft{Pos}}

\newcommand{\Lin}{\setft{L}}

\DeclareMathOperator{\poly}{poly}
\DeclareMathOperator{\negl}{negl}
\newcommand{\dset}{G}

\newcommand{\Gen}{\textsc{Gen}}

\newcommand{\eps}{\varepsilon}

\newcommand{\aux}{\textsc{aux}}

\newcommand{\ac}{\textsc{ac}}

\newcommand{\inj}{J}
\newcommand{\mZ}{\mathbb{Z}}

\newcommand{\sX}{\mathcal{X}}
\newcommand{\sY}{\mathcal{Y}}

\newcommand{\RSPV}{\ensuremath{\textsc{RSP}_{V}}}
\newcommand{\RSPB}{\ensuremath{\textsc{RSP}_{B}}}
\newcommand{\RSPS}{\ensuremath{\textsc{RSP}_{S}}}

\def\*#1{\mathbf{#1}}

\newif\ifnotes\notesfalse

\input{marginnotes}

\begin{document}

\title{Computationally-secure and composable remote state preparation}

\sloppy

\author[1]{Alexandru Gheorghiu\footnote{Email: andrugh@caltech.edu}}
\author[1]{Thomas Vidick\footnote{Email: vidick@cms.caltech.edu}}
\affil[1]{Department of Computing and Mathematical Sciences, California Institute of Technology}

\date{}
\maketitle

\noteswarning

\begin{abstract}
We introduce a protocol between a classical polynomial-time verifier and a quantum polynomial-time prover that allows the verifier to securely delegate to the prover the preparation of certain single-qubit quantum states. The protocol realizes the following functionality, with computational security: the verifier chooses one of the observables $Z$, $X$, $Y$, $(X+Y)/\sqrt{2}$, $(X-Y)/\sqrt{2}$; the prover receives a uniformly random eigenstate of the observable chosen by the verifier; the verifier receives a classical description of that state. The prover is unaware of which state he received and moreover, the verifier can check with high confidence whether the preparation was successful.  

The delegated preparation of single-qubit states is an elementary building block in many quantum cryptographic protocols. We expect our implementation of  ``random remote state preparation with verification'' (\RSPV), a functionality first defined in (Dunjko and Kashefi 2014), to be useful for removing the need for quantum communication in such protocols while keeping functionality. 

The main application that we detail is to a protocol for blind and verifiable delegated quantum computation (DQC) that builds on the work of (Fitzsimons and Kashefi 2018), who provided such a protocol with quantum communication. 
Recently, both blind an verifiable DQC were shown to be possible, under computational assumptions, with a classical polynomial-time client (Mahadev 2017, Mahadev 2018). 
Compared to the work of Mahadev, our protocol is more modular, applies to the measurement-based model of computation (instead of the Hamiltonian model) and is composable.
Our proof of security builds on ideas introduced in (Brakerski et al.\ 2018).
\end{abstract}

\newpage
\tableofcontents
\newpage

\section{Introduction}
\label{sec:intro}
\input{intro}

\section{Preliminaries}
\label{sec:prelim}
\input{prelim}

\section{Remote state preparation: real protocol}
\label{sec:real-protocol}
\input{protocol}

\section{Remote state preparation: ideal functionality}
\label{sec:ideal-protocol}
\input{composable}

\section{Blind and verifiable computation from remote state preparation}
\label{sec:dqc}
\input{fk}

\appendix

\input{ntcf}

\bibliography{rspc}

\notesendofpaper

\end{document}

%% file: marginnotes.tex
\newcommand{\tnote}[1]{\textcolor{magenta}{\small {\textbf{(Thomas:} #1\textbf{)
      }}}}
\newcommand{\anote}[1]{\textcolor{red}{\small {\textbf{(Anand:} #1\textbf{) }}}}
\newcommand{\gnote}[1]{\textcolor{orange}{\small {\textbf{(Andru:} #1\textbf{) }}}}
\newcommand{\ftnote}[1]{\footnote{\textcolor{magenta}{\small {\textbf{(Thomas:} #1\textbf{) }}}}}
\newcommand{\tdnote}[1]{\textcolor{blue}{\small {\textbf{(TODO:} #1\textbf{) }}}}

\ifnotes
\usepackage{color}
\definecolor{mygrey}{gray}{0.50}
\newcommand{\notename}[2]{{\textcolor{mygrey}{\footnotesize{\bf (#1:} {#2}{\bf ) }}}}
\newcommand{\noteswarning}{{\begin{center} {\Large WARNING: NOTES ON}\endnote{Warning: notes on}\end{center}}}
\newcommand{\notesendofpaper}{{\theendnotes}}
\newcommand{\pnote}[1]{{\endnote{#1}}}

\else

\newcommand{\notename}[2]{{}}
\newcommand{\noteswarning}{{}}
\newcommand{\notesendofpaper}{}
\newcommand{\pnote}[1]{}

\renewcommand{\tnote}[1]{}
\renewcommand{\anote}[1]{}
\renewcommand{\gnote}[1]{}
\renewcommand{\ftnote}[1]{}
\renewcommand{\tdnote}[1]{}

\fi

%% file: intro.tex
In the problem of delegated computation a user (often referred to as \emph{client} or \emph{verifier}) is provided as input a pair $(C,x)$ of a circuit $C$ and an input $x$ for the circuit. The verifier's task is to evaluate $C(x)$ as efficiently as possible. For this the verifier may delegate some or all of the computation to a powerful but untrusted server (often referred to as the \emph{prover}). 
Let $n$ be the length of $x$ and $T$ the size of the circuit $C$. Ideally, the runtime of the verifier is (quasi-)linear in $n$ and poly-logarithmic in $T$, while the runtime of the prover is quasi-linear in $T$. (Reducing space usage, for both the verifier and the prover, is also of interest, but for simplicity we focus on time.)

A productive line of research in complexity and cryptography has led to protocols for delegated computation with increasing efficiency and whose soundness can be \emph{information-theoretic} \cite{goldwasser2015delegating} or based on \emph{cryptographic assumptions}~\cite{kilian1992note,kalai2014delegate}. The latter type include protocols utilizing public-key cryptography and making standard cryptographic assumptions, such as~\cite{holmgren2018delegating}, as well as non-interactive protocols based on more non-standard assumptions, such as~\cite{gennaro2013quadratic}.
In addition to the natural applications in cloud and distributed computing, research in delegated computation is motivated by cryptographic applications (such as short zero-knowledge proofs~\cite{groth2010short,ben2013snarks}) and connections to complexity theory (such as the theory of multiprover interactive proof systems~\cite{kalai2014delegate} and probabilistically checkable proofs~\cite{goldwasser2015delegating}). 

In this paper we are concerned with the problem of \emph{delegating quantum computations} (DQC). Here the verifier is provided as input the classical description of a quantum circuit $C$, as well as a classical input $x$ for the circuit, and its goal is to obtain the result of a measurement of the output qubit of $C$ in the computational basis, when it is executed on $x$.\footnote{For simplicity we restrict to circuits that take classical inputs and return a single classical output bit obtained as the result of a measurement that is promised to return a particular value, $0$ or $1$, with probability at least $\frac{2}{3}$. This setting corresponds to delegating \emph{decision problems}, i.e. problems in which the output is a single bit. Our results also apply to the setting of \emph{relational} or \emph{sampling} problems for which the output consists of multiple bits.} In this context the main question is the following: What security guarantees can DQC protocols achieve, and at what cost?

To gain an understanding of the current landscape around this question we briefly discuss the most relevant known results, referring to~\cite{gheorghiu2018verification} for a more extensive treatment. First we note that DQC protocols come with two related but seemingly independent types of security guarantee: \emph{blindness} and \emph{verifiability}. 
A DQC protocol is said to be blind if throughout the interaction the prover does not learn anything about the delegated computation except for an upper bound on its size.
A DQC protocol is said to be verifiable if it is unlikely for the prover to succeed in convincing the verifier to accept a false statement.
The question of {blind} delegation of quantum computation was first considered by Childs~\cite{childs2001saq}, who gave such a protocol with quantum communication.
{Verifiable} delegation of quantum computation was formalized in~\cite{abe2008,broadbent2009universal} (see also~\cite{abem,fitzsimons2017unconditionally}); the authors gave protocols for verifiable DQC, and just like Childs' protocol, these protocols also require quantum communication.

Next we consider the question of efficiency of DQC protocols, focusing on the amount of quantum communication required as a measure of the verifier's ``quantum effort''. A first class of protocols, such as those from~\cite{abe2008,broadbent2009universal}, are known as \emph{prepare-and-send} protocols. This is because the verifier is required to prepare a number of small quantum states and send them to the prover. In~\cite{abe2008} the size of these quantum states (i.e.\ the number of qubits) depends on the protocol's soundness (the probability that the verifier accepts an incorrect outcome). In~\cite{fitzsimons2017unconditionally} the verifier is only required to prepare a number of single-qubit states that depends on the protocol's soundness. A second class of protocols is \emph{receive-and-measure} protocols such as~\cite{hayashi2015verifiable,fitzsimons2018post}, in which the verifier receives single qubits from the prover and is required to measure them in one of a small number of possible bases. The protocol that requires the least quantum capability from the verifier is the one from~\cite{fitzsimons2018post}; in their protocol, the verifier only needs to measure the single qubits it receives one at a time in one of two bases, computational and Hadamard.
The most communication-efficient protocols fall in the prepare-and-send category and require a total amount of quantum communication that scales as $O(T\log(1/\delta))$ where $\delta$ is the soundness error~\cite{kashefi2017optimised}; the most efficient protocols in the second category have a cubic dependence on $T$. 

All the aforementioned protocols provide information-theoretic security (for either blindness or verifiability), and all require some limited but nonzero quantum capability for the verifier. In a recent breakthrough Mahadev introduced the first entirely classical protocol for DQC~\cite{mahadev2018classical}. The protocol operates in the \emph{Hamiltonian model} of quantum computation, in which instead of directly performing the computation $C$ the prover encodes the outcome of $C$ in the smallest eigenvalue of a local Hamiltonian $H_C$.\footnote{If the circuit returns $0$ with probability at least $\frac{2}{3}$, the smallest eigenvalue is smaller than a threshold $a$, and if it returns $0$ with probability less than $\frac{1}{3}$, the smallest eigenvalue is larger than a threshold $b>a$ (this is generally referred to as the ``Kitaev circuit-to-Hamiltonian construction''~\cite{kitaev2002classical}).} The goal of the protocol is for the prover to provide evidence that it has prepared an eigenstate $\ket{\psi}$ of $H_C$ with associated eigenvalue strictly smaller than $a$. At the heart of Mahadev's result is a commitment procedure that allows the prover to commit to individual qubits of $\ket{\psi}$, and subsequently reveal a measurement outcome for a basis of the verifier's choice, {using classical communication alone}. 

The fact that the verifier in Mahadev's protocol is entirely classical marks a major departure from previous works, yet it comes at a cost in terms of security and efficiency. 
The security of the protocol is computational and rests on the post-quantum security of the learning with errors problem (LWE); moreover, the protocol is not blind, as the circuit has to be communicated to the prover so that it can determine $H_C$ and prepare an eigenstate.\footnote{The protocol can in principle be made blind by combining it with a scheme for \emph{quantum homomorphic encryption}~\cite{mahadev2018classicalQHE} but this introduces yet another layer of complexity.}
In terms of efficiency, the transformation from circuit to Hamiltonian results in an eigenvalue estimation problem that needs to be solved with accuracy at least $b-a = O(1/T^2)$ for the best constructions known~\cite{Bausch2018analysislimitations}. As a result the prover has to prepare $\Omega(T^2)$ copies of the ground state, which implies that at least $\Omega(nT^2)$ single qubits have to be sent by the prover. Moreover, preparation of a smallest eigenvalue eigenstate $\ket{\psi}$ of $H_C$ requires a circuit whose depth scales linearly with $T$, rather than with the depth of $C$. This induces a large overhead on the prover's side when the circuit $C$ has low depth but high width\footnote{Such circuits are highly parallelizable and one might hope for the complexity of delegating one to scale with depth rather than with total circuit size.}. 

Finally, and arguably most importantly, the protocol is monolithic and not obviously composable: while it solves the desired task of verification of quantum computation, it is not at first clear how or even if the protocol can be simplified to solve more elementary problems (e.g.\ verifying the preparation of a single qubit state or verifying the application of an elementary quantum operation) or combined with other cryptographic primitives (e.g.\ to remove or reduce quantum communication in a larger protocol). 

Our work is motivated by the following question: does there exist a delegation protocol for quantum computation that combines the appealing feature of having an entirely classical verifier while maintaining the relative efficiency (small polynomial overhead), simplicity (prover's computation is as close as possible to direct computation of delegated circuit), and security guarantees (verifiability, blindness, composability) of protocols with quantum communication?

\subsection{Our results}
We answer the question in the affirmative by providing an efficient, composable classical protocol for blind and verifiable DQC. The honest prover in our protocol only needs to implement the desired computation, expressed as a computation in the measurement-based model of computation, together with a sequential pre-processing phase consisting of a number of rounds that depends on the circuit size but such that the complexity of implementing each round scales only with the security parameter. The protocol combines the benefits of the best prepare-and-send quantum-verifier protocols for DQC but requires only classical communication; the downside is that our protocol is computationally sound.

Our DQC protocol is based on a basic quantum functionality that we develop and that we believe has wider applicability than the specific application to DQC. 
More precisely, we provide a computationally sound and composable protocol for the following two-party task, termed \emph{random remote state preparation} (RSP): Alice (whom we will later identify with the verifier) receives either a uniformly random bit $b \in \{0, 1\}$ or a uniformly random value $\theta\in\Theta = \{0,\frac{\pi}{4},\ldots,\frac{7\pi}{4}\}$ and Bob (whom we will later identify with the prover) receives the single-qubit state $\ket{b}$, in the case when Alice gets $b$, or the state $\ket{+_\theta}=\frac{1}{\sqrt{2}}(\ket{0}+e^{i\theta}\ket{1})$, in the case when Alice gets $\theta$. Informally, this amounts to Alice having the ability to ``steer'' a random state $\ket{+_\theta}$ (or $\ket{b}$) within Bob's workspace, using classical communication only, and such that Bob does not learn the value of $\theta$ (or $b$, respectively). (The actual functionality is slightly more complicated; see Section~\ref{sec:intro-rsp} and Figure~\ref{fig:rspv}.) 

The idea for RSP was introduced by Dunjko and Kashefi~\cite{dunjko2016blind}. The main functionality they consider is a weaker variant of RSP termed \emph{random remote state preparation with blindness}, or \RSPB.  Intuitively, the latter functionality ensures that Bob learns no information about $\theta$,
but it allows him to receive a state that is different from $\ket{+_\theta}$. The authors show that \RSPB\ (and variants of it) can be composed with a prepare-and-send protocol due to~\cite{broadbent2009universal} to achieve blind (but not verifiable) delegated computation. In~\cite{cojocaru2018delegated} a candidate implementation of \RSPB\ is given and shown secure against a limited class of adversaries referred to as  ``honest-but-curious'' adversaries. The authors of~\cite{dunjko2016blind}
 also discuss a stronger form of their primitive, called \RSPS\ (for \emph{strong}), and observe that it can be used to achieve blind and verifiable DQC by composing it with the protocol of~\cite{fitzsimons2017unconditionally}. 
The authors do not, however, provide any instantiation of \RSPS\ (other than the trivial one, using quantum communication). 

Our main contribution is to define an ideal functionality, denoted \RSPV\ (random remote state preparation with verification),\footnote{It is not hard to verify that \RSPV\ is functionally equivalent to \RSPS, in the sense that either functionality can be used to implement the other using a simple protocol. Since the definitions are syntactically different, we use a different name to avoid confusion.} and show that it can be implemented using a protocol having computational security and classical communication (see Theorem~\ref{thm:brsp.rspv} for a formal statement and Section~\ref{sec:intro-rsp} for a definition of \RSPV):

\begin{theorem}[Informal]
Assuming the learning with errors problem is computationally intractable for efficient quantum algorithms, there exists a protocol with classical communication that implements the functionality \RSPV\ within distance $\varepsilon > 0$ and which has $O(1/\eps^3)$ communication complexity. 
\end{theorem}

\noindent Here, by ``implements within distance $\eps$'', we mean that any efficient quantum circuit has advantage at most $\eps$ in distinguishing the real remote state preparation protocol from the ideal functionality \RSPV. 

 To show this result we introduce a protocol for remote state preparation and show that it is secure based on the learning with errors problem.
This is achieved by building on ideas from~\cite{brakerski2018cryptographic,mahadev2018classical} as well as from the literature on rigidity, self-testing, and quantum random access codes. Since this is our main technical contribution, we explain the protocol in more detail in Section~\ref{sec:intro-protocol} below. 

We view \RSPV\ as a fundamental resource for the construction of interactive protocols that involve classical communication between classical and quantum parties. For our result to be as widely applicable as possible, we establish security of our protocol in the \emph{abstract cryptography} (AC) framework~\cite{maurer2011abstract}. This allows one to use the primitive as a building block in other protocols. 

As a specific example of the versatility of RSP we obtain a new protocol for DQC that only requires classical communication.
The most natural protocol to which our construction applies is the delegated computation protocol from~\cite{fitzsimons2017unconditionally}. As already observed in~\cite{dunjko2016blind}, having a remote state preparation functionality immediately yields a blind and verifiable protocol for DQC with classical communication and computational soundness (we explain this in more detail in Section~\ref{sec:intro-del}). The resulting protocol is more ``direct'' than the Mahadev protocol, in the sense that in our construction the operations that the prover has to perform are closer to the quantum computation that the verifier is delegating. (The protocol from~\cite{fitzsimons2017unconditionally} operates in the {measurement-based quantum computing model},\footnote{It should be noted that the translation from the circuit model to MBQC incurs only a linear increase in overhead and this is also true for the protocol from~\cite{fitzsimons2017unconditionally}, as explained in~\cite{kashefi2017optimised}.} but we expect that protocols in the circuit model such as~\cite{broadbent2018verify} can also be implemented from \RSPV; see Section~\ref{sec:intro-del} for a discussion.)

If one assumes that \RSPV\ can be implemented at unit cost then the protocol we obtain is also more efficient than Mahadev's: for fixed soundness error, $\delta$, the number of operations performed by the prover scales linearly in the size of the delegated circuit and polynomially in the security parameter of the protocol. Unfortunately, our current version of RSP does not have unit cost. Furthemore, the number of uses of RSP required is linear in the circuit size, $T$. This implies that each use must be implemented with error $O(\delta/T)$. With our current analysis, assuming we take $\delta$ to be a constant, this results in a total communication that scales as $O(T^4)$ (see Section~\ref{sec:dqc} for a more fine-grained analysis). This is not as good as the quasi-linear complexity of prepare-and-measure protocols that use quantum communication. It is important to note, however, that the added overhead of the protocol stems from \RSPV. Thus, any improvement in the complexity of doing the state preparation will lead to an improvement in the complexity of the resulting DQC protocol. We believe reducing the overhead of \RSPV\ is possible and mention a potential way of achieving this in Section~\ref{sec:intro-del} below. We also note that our protocol consists of a sequence of simple tests that play a similar role to the Bell test in multi-prover entanglement-based protocols for DQC~\cite{reichardt2013classical}. The protocol allows for a constant fraction of failed tests, so that a partially faulty device may in principle be used to implement the protocol successfully.

Before proceeding with more details of our approach it may be useful to briefly address the following question: can one use the protocol from~\cite{mahadev2018classical} directly to implement RSP? Specifically, couldn't one enforce that the prover prepares a small-eigenvalue eigenstate of the Hamiltonian $H_\theta = -\proj{+_\theta}$? In fact it is not at all straightforward to do this.
The reasons are related to aspects of the Mahadev protocol discussed earlier. First, the committment procedure results in a state that can be measured in one of two possible bases, but it is not clear if any other form of computation besides a direct measurement can be performed on the committed qubit. Second, the guarantee provided is only that the state ``exists'' (i.e.\ the Hamiltonian has a small-eigenvalue eigenstate), but not that the state has actually been prepared by the prover. Finally, the information that the prover may have about the state it prepared is not explicitly limited (in the protocol from~\cite{mahadev2018classical} the prover learns a classical description of the Hamiltonian, hence, in this case, the value of $\theta$); forcing the prover to prepare an unknown state may require adding an additional layer of (quantum) homomorphic encryption to the protocol.

\subsection{Remote state preparation: ideal resource}
\label{sec:intro-rsp}

We formulate our variant of RSP as a \emph{resource} in the abstract cryptography framework~\cite{maurer2011abstract}. 
Abstract cryptography (AC), similar to universal composability (UC)~\cite{canetti2001universally}, is a framework for proving the security of cryptographic protocols in a way that ensures that the protocols can be securely composed in arbitrary ways.
Informally, the idea is to argue that a given protocol, which we refer to as the \emph{real protocol}, is indistinguishable from an \emph{ideal functionality} (or \emph{resource}) that captures precisely what honest or dishonest parties should be able to achieve in the protocol. 
This involves proving two things: \emph{correctness}, meaning that any efficient family of circuits (known as a \emph{distinguisher}) that interacts either with an honest run of the real protocol or with the ideal functionality has a negligible advantage in deciding which it is interacting with; \emph{security}, meaning that any attack that a malicious party could perform in the real protocol can be mapped to an attack on the ideal functionality. This latter property is formalized by saying that there exists an efficient family of (quantum) circuits, known as a \emph{simulator}, such that any distinguisher interacting with the ideal functionality and the simulator, or with the real protocol involving only the honest parties, has negligible advantage in deciding which it is interacting with.
Showing that such a simulator exists is usually the main difficulty in proving security in AC.
Since the existing results on the composability of DQC protocols are expressed in the AC framework, we also present our results in AC.
For more details on the framework we refer to Section~\ref{subsect:composable} and~\cite{maurer2011abstract, dunjko2014composable}. For the purposes of this introduction we assume basic familiarity with the framework.

We denote our variant of the ideal RSP by \RSPV, for \emph{random Remote State Preparation with Verification}. The name is chosen in direct analogy to the resource \RSPB\ of \emph{random Remote State Preparation with Blindness} introduced in~\cite{dunjko2016blind}. The resource \RSPV\ is represented schematically in Figure~\ref{fig:rspv}. In the resource, Alice inputs a bit $W\in\{X,Z\}$ that denotes a measurement basis, computational ($W=Z$) or Hadamard ($W=X$). Bob inputs a bit $c\in\{0,1\}$ that denotes honest ($c=0$) or malicious ($c=1$) behavior. If $c=0$ then in the case when $W=Z$ Alice receives a uniformly random bit $b\in\{0,1\}$ and Bob receives the state $\ket{b}$; in the case when $W=X$ Alice receives a uniformly random value $\theta \in \Theta = \{0,\frac{\pi}{4},\ldots,\frac{7\pi}{4}\}$ and Bob receives the state $\ket{+_\theta}$. If $c=1$ both Alice and Bob receive an $ERR$ message, indicating abort.

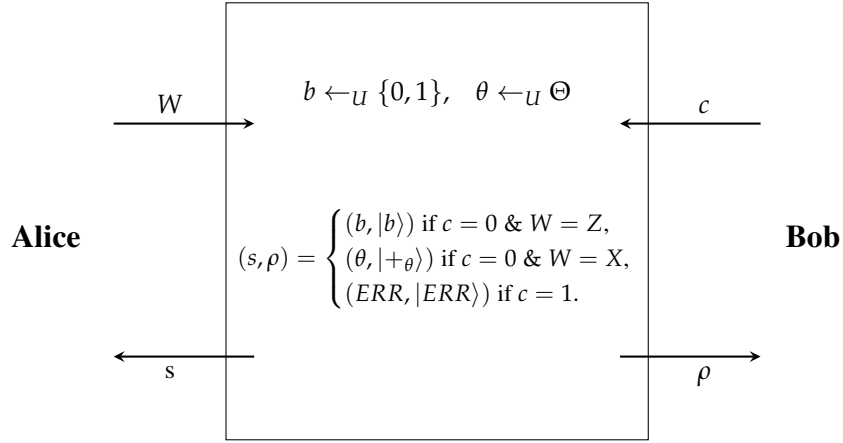
\begin{figure}[htb]
\begin{centering}

\begin{tikzpicture}[opnode/.style={minimum width=3.05cm,minimum height=.6cm}]
\small

\def\t{4.3}
\def\u{0.9}
\def\v{1}

\node[opnode] (a1) at (-\u,0) {};
\node[opnode] (a2) at (-\u,-1*\v) {};
\node[opnode] (a3) at (-\u,-2.2*\v) {};
\node[opnode] (a4) at (-\u,-3.1*\v) {};
\node (alice1) at (-\t,0) {};
\node (alice2) at (-\t,-1*\v) {};
\node (alice3) at (-\t,-2.2*\v) {};
\node (alice4) at (-\t,-3.1*\v) {};
\node (alice) at (-\t-.9,-1.5*\v) {\large \textbf{Alice}};

\node[opnode] (b1) at (\u,0 * \v) {};
\node[opnode] (b2) at (\u,-0.5 * \v) {};
\node[opnode] (b3) at (\u,-1*\v) {};
\node[opnode] (b4) at (\u,-2.2*\v) {};
\node[opnode] (b5) at (\u,-3.1*\v) {};
\node (bob1) at (\t,0*\v) {};
\node (bob2) at (\t,-0.5 * \v) {};
\node (bob3) at (\t,-1*\v) {};
\node (bob4) at (\t,-2.2*\v) {};
\node (bob5) at (\t,-3.1*\v) {};
\node (bob) at (\t+.7,-1.5*\v) {\large \textbf{Bob}};

\node[opnode] (c1) at (0,0.4) {$b \leftarrow_U  \{0, 1\}, \;\;\; \theta \leftarrow_U  \Theta$};
\node[opnode] (c2) at (0, 0.0) {};
\node[opnode] (c3) at (0,-1.2*\v) {};
\node[opnode] (c4) at (0,-1.8*\v) {\footnotesize $\;\;\; (s, \rho) = \begin{cases}
    (b, \ket{b}) \text{ if $c=0 \; \& \; W=Z$,} \\ (\theta, \ket{+_{\theta}}) \text{ if $c=0 \; \& \; W=X$,} \\
    (ERR, \ket{ERR}) \text{ if $c=1$}.
      \end{cases}$};
\node[opnode] (c5) at (0,-2.7*\v) {};
\node[opnode] (c6) at (0,-3.0*\v) {};
\node[draw,minimum width=4.9cm,inner sep=.9cm,fit=(c1)(c6)] (c) {};

\draw[sArrow] (alice1.center) to node[auto,pos=0.4] {$W$} (a1);
\draw[sArrow] (a4) to node[auto,pos=0.6] {s} (alice4.center);

\draw[sArrow] (bob1.center) to node[auto,pos=.4,swap] {$c$} (b1);
\draw[sArrow] (b5) to node[auto,pos=.6,swap] {$\rho$} (bob5.center);

\end{tikzpicture}

\end{centering}
\caption{\label{fig:rspv} The resource \RSPV. It chooses $b$ uniformly at random from $\{0, 1\}$ and $\theta$ uniformly at random from $\Theta = \{0,\frac{\pi}{4},\ldots,\frac{7\pi}{4}\}$. It takes $W\in\{X,Z\}$ as input from Alice and $c\in\{0,1\}$ as input from Bob. When $c=0$ it outputs either $b$ to Alice and $\ket{b}$ to Bob, if $W=Z$; or $\theta$ to Alice and $\ket{+_{\theta}}$ to Bob, if $W=X$. When $c=1$ it outputs $ERR$ to Alice and $\ket{ERR}$ to Bob.}
\end{figure}

Note that the resource \RSPV\ can almost be understood as a communication channel from Alice to Bob that would allow Alice to select one of $10$ possible single-qubit states $\ket{0},\ket{1}$, or $\ket{+_\theta}$ for $\theta\in \Theta$ and send it to Bob. There are two differences: first, Alice does not choose the state, but instead the functionality chooses it uniformly at random and tells Alice what it is. Second, Bob may decide to block the channel, in which case both parties receive an error message.  This in contrast with the weaker resource of \RSPB, also introduced in~\cite{dunjko2016blind} and for which~\cite{cojocaru2018delegated} give a real protocol with security against ``honest-but-curious'' adversaries, in which Bob is allowed to select the family of states $\{\rho_\theta\}$ that it receives (by explicitly specifying them to the resource).\footnote{The $\rho_\theta$ should satisfy the consistency condition shown in~\eqref{eq:rho_consistency}, which says that it is possible to generate the state $\rho_\theta$ by performing a $\theta$-dependent measurement on a fixed state $\rho$; we refer to~\cite{dunjko2016blind} for details.} The resource \RSPV\ allows less flexibility to a dishonest user, making it more useful as a building block. In particular, the rigidity of Bob's output state is essential to obtain a protocol that is verifiable.

\subsection{Remote state preparation: real protocol}
\label{sec:intro-protocol}

In the previous section we defined the ideal functionality for remote state preparation with verifiability, \RSPV. In this section we describe a protocol that we prove is computationally indistinguishable from the ideal functionality. The protocol builds on ideas from~\cite{brakerski2018cryptographic} and~\cite{mahadev2018classical}. The main difficulty in the implementation of \RSPV\ is to obtain \emph{verifiability}, i.e.\ the guarantee that an \emph{arbitrary} (computationally bounded) prover successfully interacting with the verifier \emph{must} have prepared locally the correct state, and yet have obtained no more information (computationally) about the state itself than could be gained had the state been sent directly by the verifier (or the ideal resource). To achieve this we significantly strengthen the rigidity argument from~\cite{brakerski2018cryptographic} by giving more control, and freedom, to the verifier in the kinds of states that are prepared. 

In the real protocol, that we call the \emph{buffered remote state preparation protocol} (BRSP), Alice and Bob interact through two communication resources: a classical channel as well as a \emph{measurement buffer}. The measurement buffer takes as input a classical message $M$ from Alice, and from Bob a specification (as a quantum circuit) of a measurement for each of the possible messages of Alice, as well as a state on which the measurement is to be performed (as a quantum state). The buffer then performs the measurement associated with Alice's message, forwards the outcome to Alice, and returns the post-measurement state to Bob. 

\begin{figure}[htb]
\begin{centering}

\begin{tikzpicture}[opnode/.style={minimum width=1.05cm,minimum height=.4cm}]
\small

\def\t{3.5}
\def\u{0.8}
\def\v{1}

\node[opnode] (a1) at (-\u,0) {};
\node[opnode] (a2) at (-\u,-1*\v) {};
\node[opnode] (a3) at (-\u,-2.4*\v) {};
\node (alice1) at (-\t,0) {};
\node (alice2) at (-\t,-1*\v) {};
\node (alice3) at (-\t,-2.4*\v) {};
\node (alice) at (-\t-.9,-0.5*\v) {\large \textbf{Alice}};

\node[opnode] (b1) at (\u,0 * \v) {};
\node[opnode] (b2) at (\u,-0.5 * \v) {};
\node[opnode] (b3) at (\u,-1*\v) {};
\node (bob1) at (\t,0*\v) {};
\node (bob2) at (\t,-0.5 * \v) {};
\node (bob3) at (\t,-1*\v) {};
\node (bob) at (\t+.7,-0.5*\v) {\large \textbf{Bob}};

\node[opnode] (c1) at (0,0) {};
\node[opnode] (c2) at (0,-0.6*\v) {$r \leftarrow \mathcal{F}(M)(\rho)$};
\node[opnode] (c3) at (0,-1.2*\v) { };
\node[draw,minimum width=3.2cm,inner sep=.25cm,fit=(c1)(c3)] (c) {};
\node[yshift=-2,above right] at (c.north west) {};

\draw[sArrow] (alice1.center) to node[auto,pos=0.4] {$M$} (a1);
\draw[sArrow] (a2) to node[auto,pos=0.4] {$r$} (alice2.center);

\draw[sArrow] (bob1.center) to node[auto,pos=.4,swap] {$\mathcal{F}$} (b1);
\draw[sArrow] (bob2.center) to node[auto,pos=.4,swap] {$\rho$} (b2);
\draw[sArrow] (b3) to node[auto,pos=.65,swap] {$r, [\mathcal{F}(M)(\rho)]_r$} (bob3.center);

\end{tikzpicture}

\end{centering}
\vspace{-0.3in}
\caption{\label{fig:meas.buffer} The measurement buffer. Alice inputs a message $M$. Bob inputs a specification $\mathcal{F}$ which takes as input Alice's message and returns a measurement $\mathcal{F}(M)$. Bob also inputs a state $\rho$. The buffer measures $\rho$ with $\mathcal{F}(M)$ producing classical outcome $r$ and the post-measurement state denoted $[\mathcal{F}(M)(\rho)]_r$. Both Alice and Bob receive $r$ and Bob also receives $[\mathcal{F}(M)(\rho)]_r$.}
\end{figure}
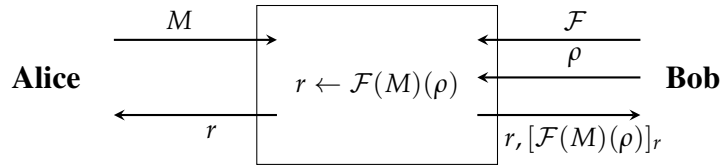

The necessity of relying on a measurement buffer to obtain a secure protocol is a consequence of the use of rigidity to obtain verifiability. Rigidity arguments require the assumption that, in an execution of the real protocol, the measurements implemented by Bob are ``local''; in other words, that the simulator constructed in the security proof can interact directly with those measurements. In the AC framework, in general, a malicious Bob may ``delegate'' any measurements that it wishes to make to the environment\footnote{In AC and UC, the environment represents anything that is external to the protocols under consideration \cite{canetti2001universally, maurer2011abstract}. This can include other protocols, other parties etc.}, which would render them inaccessible to the simulator. By constructing the protocol from a measurement buffer resource we explicitly prevent such behavior from Bob. (Note that the use of the buffer does \emph{not} prevent Bob from sharing entanglement with the environmnent, or from exchanging quantum states with it in-between any two uses of the measurement buffer.)
While the measurement buffer is necessary to obtain composable security, the use of this resource can be omitted when considering stand-alone security only (since in that case, there is no environment). 
Finally, note that the measurement buffer is not a ``physical'' resource of the protocol; in an actual run of the protocol Alice and Bob interact only classically.
The buffer is discussed in more detail in Section~\ref{sec:ideal-protocol}.

We proceed with an informal description of the protocol and its analysis. Our starting point is the work~\cite{brakerski2018cryptographic}, in which the authors give a classical protocol between a verifier and prover such that provided the prover is accepted with non-negligible probability in the protocol, it is guaranteed that a subset of the values returned by the prover contain information-theoretic randomness. This guarantee holds as long as the prover is computationally bounded, and more specifically that it does not have the ability to break the learning with errors (LWE) problem while the protocol is being executed. 

We observe that the proof of~\cite{brakerski2018cryptographic} explicitly establishes a stronger rigidity statement whereby the prover is guaranteed, up to a local rotation on its workspace, to have prepared a $\ket{+}$ state and measured it in the computational basis (hence the randomness). Formulated differently, the protocol from~\cite{brakerski2018cryptographic} implements a weak variant of \RSPV\ in which only the option $W=Z$ is available to Alice. This is not sufficient for delegated computation, but it is a starting point. 

To generate the other states needed for \RSPV\ we need to go deeper in the protocol from~\cite{brakerski2018cryptographic}. At a high level, the idea is to engineer the preparation of a state of the form 
\begin{equation}\label{eq:psi0}
 \frac{1}{\sqrt{2}} \big(\ket{0} \ket{x_0} + \ket{1}\ket{x_1} \big)\;,
\end{equation}
where $x_0,x_1\in\{0,1\}^w$ are bitstrings defined as the unique preimages of an element $y$, provided by the prover to the verifier, under a claw-free pair of functions $f_0,f_1:\{0,1\}^w \to \mY$, where $\mY$ is some finite range set. For the purposes of this discussion it is not important how the state~\eqref{eq:psi0} is obtained, as long as we can guarantee that the prover prepares such a state. 

In~\cite{brakerski2018cryptographic} the next step is to ask the prover to measure the second register in the Hadamard basis (i.e.\ implement the Fourier transform over $\Z_2^w$ and then measure in the computational basis). Labeling the outcome as $d\in\{0,1\}^w$, the first qubit is projected to the state $\frac{1}{\sqrt{2}}(-1)^{d\cdot x_0}(\ket{0} + (-1)^{d\cdot(x_0\oplus x_1)} \ket{1})$ that provides the basis for the randomness generation described earlier. 

Consider the following simple modification: by thinking of $x_0,x_1$ as elements of $\Z_8^{w/3}$ (assuming $w$ is a multiple of $3$) instead of $\{0,1\}^w$, we can ask the prover to implement the Fourier transform over $\Z_8$, yielding an outcome $d\in \Z_8^{w/3}$ and a post-measurement state 
 \begin{equation}\label{eq:psi1}
\ket{\psi_\theta}\,=\,\frac{1}{\sqrt{2}} \omega^{d\cdot x_0}(\ket{0} + \omega^{d\cdot(x_0+ x_1)} \ket{1})\;,
\end{equation}
where $\omega = e^{\frac{2i\pi}{8}}$ and the addition and inner product are taken modulo $8$. Up to a global phase this is precisely the state $\ket{+_\theta}$, for $\theta = \frac{\pi}{4}\, d\cdot (x_0+ x_1)$. 

So far the argument establishes completeness: if Alice and Bob follow the protocol, Alice obtains an angle $\theta$ and Bob obtains the state $\ket{+_\theta}$. Moreover, using a slight extension of the adaptive hardcore bit statement from~\cite{brakerski2018cryptographic} it is not hard to show that the value of $\theta$ is computationally indistinguishable from uniform from Bob's perspective. The main difficulty is to argue that the prover \emph{must} have created \emph{precisely} the state $\ket{\psi_\theta}$ in~\eqref{eq:psi1}, and not for instance a related state such as $\ket{\psi_{3\theta}}$. Note that this would be allowed in \RSPB, but it is not in \RSPV.

In order to show that the prover must have a state that is equal, up to an isometry, to a state of the form $\ket{\psi_\theta}$ we combine rigidity arguments similar to those employed in~\cite{brakerski2018cryptographic} with a new idea: we introduce a test that asks the prover to demonstrate that the state it has prepared implements a near-optimal $2\mapsto 1$ quantum random access code (QRAC). A $2\mapsto 1$ QRAC is a procedure that encodes two classical bits into a single qubit, in a way that maximizes the success probability of the following task: given a request for either the first or the second bit (chosen with equal probability), perform a measurement on the single qubit that returns the value of that bit with the highest possible probability. As shown in~\cite{ambainis2008quantum} the optimum success probability of this task is $\frac{1}{2} + \frac{1}{2\sqrt{2}}$, and is achieved by encoding the two bits in one of the four single-qubit states $\ket{+_0},\ket{+_{\frac{\pi}{2}}},\ket{+_\pi}$ and $\ket{+_{\frac{3\pi}{2}}}$. More specifically, if the input bits are denoted $b_1$, $b_2$, then the QRAC state is $\ket{+_{b_1 \pi + b_2 \frac{\pi}{2}}}$. Moreover, the optimal measurement for predicting one bit or the other is a measurement in the basis $\{\ket{+_{\frac{\pi}{4}}}, \ket{+_{\frac{5\pi}{4}}}\}$, if $b_1$ is requested, or $\{\ket{+_{\frac{3\pi}{4}}},\ket{+_{\frac{7\pi}{4}}}\}$, if $b_2$ is requested. 

We extend the optimality proof from~\cite{ambainis2008quantum} to show that even a near-optimal family of states and measurements must be close, up to a global rotation, to the ones described above. Next we enforce that the prover's states and measurements implement a near-optimal $2\mapsto 1$ QRAC by asking that the prover successfully predict certain bits of $\theta$, given partial information about it. For example, the verifier can reveal to the prover that $\theta \in \{\frac{\pi}{2},\frac{3\pi}{2}\}$ and ask which is the case; the prover should be able to answer with probability $1$ by performing the appropriate measurement. Or the verifier can reveal that $\theta\in \{\ket{+_0},\ket{+_{\frac{\pi}{2}}},\ket{+_\pi},\ket{+_{\frac{3\pi}{2}}}\}$ and ask the prover to guess one additional bit of $\theta$; the prover should be able to succeed with probability $\frac{1}{2} + \frac{1}{2\sqrt{2}}$. 

Making use of the rigidity argument to establish composable security requires the simulator to have access to Bob's measurement operators. For this reason, while most communication steps of the protocol can be implemented using a classical communication channel, in the last step of the protocol, described in the previous paragraph, the communication takes place through a measurement buffer: Alice inputs partial information about $\theta$, and Bob inputs a description of the measurement that he would have performed on each of Alice's possible questions, together with the quantum state on which the measurement is to be performed. 

The complete argument is given in Section~\ref{sec:real-protocol}. We introduce a sequential protocol that consists of a number $N$ of tests, followed by a random stopping time. We show that any behavior of the prover that has non-negligible probability of passing a fraction of tests that is within a small enough constant of the optimal fraction is such that the following property holds: at the end of the protocol, the state of the prover is unitarily equivalent to a state that is computationally indistinguishable (up to a small computational error that depends on $N$ and other parameters of the protocol) from a state of the form $\ket{\psi_\theta}$ together with some $\theta$-independent side information.

\subsection{Application: delegated computation}
\label{sec:intro-del}

Having defined the ideal \RSPV\ functionality as well as the real protocol that implements this functionality from classical channels, we now discuss applications. As mentioned, the most natural application of RSP is to verifiable delegated quantum computation. 
Intuitively, the idea is the following: suppose Alice wishes to delegate $C(x)$ to Bob, for some quantum circuit $C$ having $T$ gates. Using the measurement-based protocol from~\cite{fitzsimons2017unconditionally}, if Alice were to send Bob $O(T \, \log(1 / \delta))$ randomly chosen states, from the ten possible choices mentioned earlier (the $\ket{+_{\theta}}$ states, with $\theta \in \Theta$, and the $\ket{0}$, $\ket{1}$ states), she would be able to delegate $C(x)$ to Bob and the protocol would have soundness error at most $\delta$.
The \RSPV\ functionality allows her to do exactly this, using classical communication alone. Of course, unlike the protocol of~\cite{fitzsimons2017unconditionally}, the security of this construction would be computational, rather than information-theoretic.
To summarize, in the delegation protocol Alice first executes \RSPV\ a certain number of times with Bob in order to prepare the required resource states in Bob's quantum memory. She then engages in the protocol of~\cite{fitzsimons2017unconditionally} as if she had sent the random states to Bob.

How many times does Alice need to execute \RSPV? To delegate the circuit of size $T$ and achieve soundness error $\delta$, the number of executions must clearly be at least $\Omega(T \, \log(1 / \delta))$. If the real protocol used to implement \RSPV\ prepared the intended states \emph{exactly}, then we would have exactly that many runs. Of course, this is not the case, and we need to account for the failure probability of the real protocol, which we denote as $\varepsilon$. It was shown in~\cite{dunjko2014composable, gheorghiu2015robustness} that the protocol of~\cite{fitzsimons2017unconditionally} is robust to deviations in the collective state of the resource qubits. If there are $M$ such qubits, and the error per state is $\varepsilon$ then by the triangle inequality it follows that the deviation of the whole state is at most $M \varepsilon$. We therefore need to choose $\varepsilon = O(\delta/M)$ and since $M = \Omega(T \, \log(1 / \delta))$, this means that $\varepsilon = O \left( \frac{\delta}{T \, \log(1 / \delta)} \right)$. As shown in Section~\ref{sec:real-protocol}, to achieve error at most $\varepsilon$, the real protocol associated to \RSPV\ must have a running time of $O(1/\varepsilon^3)$.
Putting everything together, this leads to a total number of operations that scales as $O((T^4/\delta^3) \, \log^4(1/\delta))$. Ideally, one may hope for an implementation whose communication is linear in $T$. It may be possible to do this by considering a single-use parallel version of our protocol, whereby all states would be generated in a single iteration. Achieving this is likely to be technically challenging, and we leave the possibility open for future work. 

In the language of AC, the ideal functionality for verifiable DQC has already been defined in~\cite{dunjko2016blind}. What we show is that this functionality is computationally indistinguishable from the real protocol described earlier. To do this we first adapt the definitions of DQC resources to the setting of computational security. We then show that the results pertaining to those resources in the information-theoretic case also hold in the case of computational security. This is done in Section~\ref{sec:dqc}. Finally, we show that the \RSPV\ functionality can be used to implement the computational DQC functionalities. It follows that the real protocol we described is computationally indistinguishable from the ideal DQC resource.

As already mentioned one of the main advantages to proving the security of \RSPV\ in the AC framework is that one can directly plug this primitive into other existing protocols. Aside from DQC, a related application is to \emph{multi-party quantum computation} (MPQC). In~\cite{kashefi2017multiparty} the authors define AC functionalities for multi-party quantum computation.  Their protocol consists of a number of clients, each having its own input, that wish to delegate a computation on their collective inputs to a quantum server. Its security, as defined in~\cite{kashefi2017multiparty}, is guaranteed in the settings where either the server is malicious (but the clients are not), or a subset of clients is malicious (but the server behaves honestly). The protocol works by having the clients perform a remote state preparation protocol, in which the clients send quantum states to the server. It then proceeds in a manner similar to the single-client DQC protocols. In principle, remote state preparation could be replaced with our \RSPV\ primitive, leading to an MPQC protocol in which the clients and the server use only classical communication.
We leave the formalization of this intuition to future work.\\

Upon completion of this work we became aware of the independent work ``QFactory: classically-instructed remote secret qubits preparation'' by Cojocaru, Colisson, Kashefi and Wallden. Using our terminology, their main result is the design of a protocol for \RSPB, the blind variant of RSP, that they prove computationally secure.

\paragraph{Outline.} We start with Section~\ref{sec:prelim}, which contains the preliminaries for this work. Most notably, in this section we recast some of the definitions pertaining to composability of DQC protocols, expressed in the AC framework, in the setting of computational, rather than information-theoretic security.
Then, in Section~\ref{sec:real-protocol}, we describe the remote state preparation protocol and prove the rigidity statement about its functionality. In other words we show that, provided the verifier accepts with non-negligible probability, the prover's state is close (up to an isometry) to the ideal random state that the verifier receives a description of. In the proof, we make use of an extended noisy trapdoor claw-free function family, for which we provide the relevant definitions in Appendix~\ref{sec:entcf}, as well as present the properties of these functions that we require.
Next, in Section~\ref{sec:ideal-protocol} we describe the ideal \RSPV\ functionality and prove, in the AC framework, that the protocol from Section~\ref{sec:real-protocol} implements this functionality from classical channels, under computational assumptions.
Having done this, we end in Section~\ref{sec:dqc} by showing that the ideal \RSPV\ functionality can be used to implement the functionality for blind and verifiable delegated quantum computation. From the previous results, this implies that one can have a computationally secure DQC protocol by using our RSP primitive to prepare the quantum states used by that protocol. The specific DQC protocol we consider is the one from~\cite{fitzsimons2017unconditionally}.

\paragraph{Acknowledgments.} 
We thank Rotem Arnon-Friedman, Vedran Dunjko, Urmila Mahadev and Christopher Portmann for useful discussions. 
Alexandru Gheorghiu and Thomas Vidick are supported by MURI Grant FA9550-18-1-0161 and the IQIM, an NSF Physics Frontiers Center (NSF Grant PHY-1125565) with support of the Gordon and Betty Moore Foundation (GBMF-12500028). Thomas Vidick is also supported by NSF CAREER Grant CCF-1553477, AFOSR YIP award number FA9550-16-1-0495, and a CIFAR Azrieli Global Scholar award.

%% file: prelim.tex
\subsection{Notation}

We write $\mH$ for a finite-dimensional Hilbert space, using indices $\mH_A$, $\mH_B$ to specify distinct spaces. $\Lin(\mH)$ is the set of linear operators on $\mH$. We write $\Id_A\in \Lin(\mH_A)$ for the identity operator,  $\Tr(\cdot)$ for the trace, and $\Tr_B:\Lin(\mH_A \otimes \mH_B )\to \Lin(\mH_A)$ for the partial trace. $\Pos(\mH)$ is the set of positive semidefinite operators and $\Density(\mH)=\{X\in \Pos(\mH):\Tr(X)=1\}$ the set of density matrices (also called states).  

Given $A \in \Lin(\mH)$, $\|A\|_1=\Tr\sqrt{A^\dagger A}$ is the Schatten $1$-norm and $TD(A) = \frac{1}{2}\|A\|_1$ the trace distance. 

Given $X,Z\in\Lin(\mH)$ we write $\{X,Z\}=XZ+ZX$ for the anticommutator and $[X,Z]=XZ-ZX$ for the commutator. $\sigma_X, \sigma_Y, \sigma_Z\in\Lin(\C^2)$ are the single-qubit Pauli matrices. For an angle $\theta$ we let $\sigma_{X,\theta} = \cos\theta\,\sigma_X + \sin \theta\, \sigma_Y$. 

A completely positive trace-preserving (CPTP) map $\mF:\mH_A\to\mH_B$ is a linear map such that for any $\mH_C$ and $\rho \in \Pos(\mH_A\otimes \mH_C)$ it holds that $(\mF\otimes \Id_C)(\rho) \in\Pos(\mH_B\otimes \mH_C)$ and $\Tr(\mF\otimes \Id_C(\rho))=\Tr(\rho)$.

We let $\Theta = \{0,\frac{\pi}{4},\frac{2\pi}{4},\ldots,\frac{7\pi}{4}\}$. For $\theta\in\Theta$, $\ket{+_\theta} = \frac{1}{\sqrt{2}}(\ket{0} + e^{i\theta}\ket{1})$. We often identify elements of $\Z_8$ with $\{0,1,2,\ldots,7\}$. For a finite set $S$ we write $x\leftarrow_U S$ to mean that $x$ is chosen uniformly at random from $S$. A negligible function is a function $\delta:\N \to \R$ that goes to $0$ faster than any inverse polynomial, i.e.\ $p(\lambda)\delta(\lambda)\to_{\lambda\to\infty} 0$ for any polynomial $p$.

\subsection{Efficient states and operations}

\begin{definition}
We say that a family of states $\{\rho_\lambda \in \Density(\mH_{A_\lambda})\}_{\lambda\in\N}$ is efficiently preparable (or just ``efficient'') if there exists a polynomial-time uniformly generated\footnote{By ``polynomial-time uniformly generated we mean that there exists a Turing machine $T$ that on input $1^\lambda$ returns a description of the circuit $C_\lambda$ using some fixed finite universal gate set.} family of circuits $\{C_\lambda\}_{\lambda\in\N}$ acting on $\mH_{A_\lambda}\otimes \mH_{B_\lambda}$ such that
\[\forall\lambda\;,\qquad \Tr_{B_\lambda}\big(C_\lambda(\proj{0}_{A_\lambda} \otimes \proj{0}_{B_\lambda})\big) = \rho_\lambda\;.\]
\end{definition}

\begin{definition}
We say that a family of CPTP maps $\{\mF_\lambda: \Lin(\mH_{A_\lambda}) \to \Lin(\mH_{B_\lambda})\}_{\lambda\in\N}$ is efficient if there exists a polynomial-time uniformly generated family of circuits $\{C_\lambda\}$ acting on $\mH_{A_\lambda}\otimes \mH_{B_\lambda}\otimes\mH_{C_\lambda}$ such that
\[\forall\lambda\;,\forall \rho\in \Density(\mH_{A_\lambda})\;,\qquad \Tr_{A_\lambda C_\lambda}\big(C_\lambda(\rho\otimes \proj{0}_{B_\lambda C_\lambda}) \big) = \mF_\lambda(\rho_\lambda)\;.\]
\end{definition}

\subsection{Computational distinguishability}

\begin{definition}\label{def:dist} 
Given two families of (not necessarily normalized) density operators $\{\rho_\lambda\}_{\lambda \in \N}$ and $\{\sigma_\lambda\}_{\lambda \in \N}$ we say that $\rho$ and $\sigma$ are \emph{computationally distinguishable with advantage at most $\delta(\lambda)$}, and write $\rho \approx_{c,\delta} \sigma$, if for any polynomial-time uniformly generated family of circuits $\{D_\lambda\}_{\lambda\in\N}$, known as a \emph{distinguisher}, there is a $\lambda_0\in\N$ such that 
\begin{equation} \label{eq:dist-0}
\forall \lambda \geq \lambda_0\;,\qquad \frac{1}{2}\big|\Tr\big(D_\lambda^\dagger(\proj{0}\otimes\Id)D_\lambda \rho_\lambda\big) - \Tr\big(D_\lambda^\dagger(\proj{0}\otimes\Id)D_\lambda \sigma_\lambda\big) \big|\,\leq\, \delta(\lambda)\;.
\end{equation}
\end{definition}

The best $\delta(\lambda)$ in Definition~\ref{def:dist} implicitly depends on the specific polynomial bound that is placed on the size of the distinguisher. In this paper it will always be the case that $\delta(\lambda) = \delta + \negl(\lambda)$, for some constant $\delta$ and a negligible function of $\lambda$. The size of the distinguisher will affect the negligible function; the statement should be interpreted as saying that for any polynomial size bound on the distinguisher there is a negligible function of $\lambda$ such that~\eqref{eq:dist-0} holds.

\begin{lemma}\label{lem:comp-inf}
For any density operators $\{\rho_\lambda\}_{\lambda \in \N}$ and $\{\sigma_\lambda\}_{\lambda \in \N}$, $\{\rho_\lambda\}_{\lambda \in \N}$ and $\{\sigma_\lambda\}_{\lambda \in \N}$ are computationally distinguishable with advantage at most $\|\rho_\lambda - \sigma_\lambda\|_1$. 
\end{lemma}

\begin{proof}
For any $\rho, \sigma$ and $0\leq D \leq \Id$ it holds that:
\begin{equation}
\frac{1}{2}\big|\Tr\big(D_\lambda^\dagger(\proj{0}\otimes\Id)D_\lambda \rho_\lambda\big) - \Tr\big(D_\lambda^\dagger(\proj{0}\otimes\Id)D_\lambda \sigma_\lambda\big) \big|\,\leq\, \|\rho_{\lambda}-\sigma_{\lambda}\|_1
\end{equation}
\end{proof}

\begin{lemma}\label{lem:dist}
For $b\in \{0,1\}$ let $\{\rho_\lambda^b\}_{\lambda \in \N}$ and $\{\sigma_\lambda^b\}_{\lambda \in \N}$ be two families of density operators. For all $\lambda$, let $\rho_\lambda = \sum_b \proj{b} \otimes \rho_\lambda^b$ and $\sigma_\lambda = \sum_b \proj{b}\otimes \sigma_\lambda^b$. Suppose that $\{\rho_\lambda\}$ and $\{\sigma_\lambda\}$ are distinguishable with advantage at most $\delta(\lambda)$. Then for $b\in \{0,1\}$, $\{\rho_\lambda^b\}$ and $\{\sigma_\lambda^b\}$ are distinguishable with advantage at most $\delta_b(\lambda)$ where $\delta_0(\lambda),\delta_1(\lambda)$ are such that $|\delta_0(\lambda) + \delta_1(\lambda) - \delta(\lambda)|=\negl(\lambda)$.
\end{lemma}

\begin{proof}
For $b\in\{0,1\}$ fix a family of efficient distinguishers $\{D^b_\lambda\}$ for $\{\rho^b_\lambda\}$ and $\{\sigma_\lambda^b\}$ with advantage $\delta_b(\lambda)$. Then the distinguisher $D = \proj{0} \otimes D^0 + \proj{1}\otimes D^1$ is efficient and has distinguishing advantage $\delta_0(\lambda)+\delta_1(\lambda)$ for $\{\rho_\lambda\}$ and $\{\sigma_\lambda\}$. This shows that  $\delta_0(\lambda)+\delta_1(\lambda)\leq \delta(\lambda)$. Conversely, let $\{D_\lambda\}$ be an efficient distinguisher for $\{\rho_\lambda\}$ and $\{\sigma_\lambda\}$ with advantage $\delta(\lambda)$. Then 
\begin{multline*}
\big|\Tr(D_\lambda^\dagger(\proj{0}\otimes\Id)D_\lambda \rho_\lambda) - \Tr(D_\lambda^\dagger(\proj{0}\otimes\Id)D_\lambda \sigma_\lambda)\big| \\ 
\leq \sum_b \big|\Tr((D_\lambda^b)^\dagger(\proj{0}\otimes\Id)D_\lambda^b \rho_\lambda^b) - \Tr((D_\lambda^b)^\dagger(\proj{0}\otimes\Id)D_\lambda^b\sigma_\lambda^b) \big|\;,
\end{multline*}
where $D_\lambda^b$ is the efficient distinguisher that initializes an ancilla qubit to state $\proj{b}$ and then runs $D_\lambda$. 
\end{proof}

\subsection{Composable security} \label{subsect:composable}

Abstract cryptography (AC) is a framework for proving the security of protocols under composition. For example if protocols $\pi_1$ and $\pi_2$ are shown to be secure in the AC framework then their sequential composition $\pi_1 \circ \pi_2$ or parallel composition $\pi_1 | \pi_2$ is automatically secure as well. 
For an in-depth introduction to the framework of abstract cryptography specialized to the present context of two-party quantum protocols we refer to~\cite{dunjko2013composable}. Here we briefly recall the key notions and terminology.  

The actions of the two players, generally called Alice and Bob, in a two-party protocol $\pi$ are specified by a sequence of CPTP maps $\pi_A=\{\mE_i: \Lin(\mH_{AC})\to \Lin(\mH_{AC})\}_i$ and $\pi_B=\{\mF_i: \Lin(\mH_{CB})\to \Lin(\mH_{CB})\}_i$, where $A$ and $B$ are Alice and Bob's private registers respectively, and $C$ represents a communication channel. In AC the channel $C$ is modeled as a resource $\mathcal{R}$, where in general a resource is itself represented as a sequence of completely positive trace-preserving (CPTP) maps with internal memory.


In the AC framework a protocol can be thought of as a process that constructs a resource, $\mathcal{S}$, from some other resource, $\mathcal{R}$. For instance, a protocol $\pi_{AB} = (\pi_A, \pi_B)$ can construct an ideal resource for delegated quantum computation from a resource consisting of classical and quantum channels. The resource $\pi_A \mathcal{R}$ obtained by plugging in one player's strategy into the resource is another resource, itself modeled as a sequence of CPTP maps, that can be thought of as a quantum strategy as defined in~\cite{gutoski2007toward}. When both Alice and Bob follow the protocol while interacting with the resource $\mathcal{R}$ we write $\pi_{AB} \mathcal{R}$, or $\pi_A \mathcal{R} \pi_B$.\footnote{The ordering of the protocols, $\pi_A$ and $\pi_B$, has no special significance.} Note that $\pi_A \mathcal{R} \pi_B$ is again itself a resource, having input and output interfaces for both Alice and Bob. 

Since the goal in the AC framework is to show that certain resources are indistinguishable from each other, we need a notion of distinguishability of resources. Informally, two resources $\mathcal{R}_1$ and $\mathcal{R}_2$, each modeled as a sequence of CPTP maps with input and output spaces of compatible dimension, are \emph{computationally distinguishable with advantage at most $\eps$} if no efficient distinguisher $\mathcal{D}$ (itself represented as a family of efficient CPTP maps) can distinguish an interaction with $\mathcal{R}_1$ from an interaction with $\mathcal{R}_2$. Here, the distinguisher is allowed to create an initial state as input to the resource (the state can be entangled with a reference system kept by the distinguisher); then, upon having received the output of the first map, it can modify it in an arbitrary (efficient) way and input it to the second map, etc., until it is required to make an (efficient) measurement on the output of the last map (and its own reference system) in order to return a guess for the resource with which it was interacting. We write the composition of $\mathcal{D}$ and $\mathcal{R}_i$, for $i\in \{1,2\}$, as $\mathcal{D}\mathcal{R}_i$; this is a resource that takes no input and outputs a single bit. 

\begin{definition} \label{def:distadv}
Let $\eps=\eps(\lambda)\in[0,1]$ be a function of a security parameter $\lambda\in\N$, and let $\mathcal{R}_1$ and $\mathcal{R}_2$ be two resources having input and output spaces of the same dimension. 
We say that $\mathcal{R}_1$ and $\mathcal{R}_2$ have distinguishing advantage $\eps$ if for all efficient distinguishers $\mathcal{D}$ it holds that $|\Pr( \mathcal{D}\mathcal{R}_1=1)-\Pr(\mathcal{D}\mathcal{R}_2=1)|\leq \eps$. We write this as:
\begin{equation}
\mathcal{R}_1 \approx_{c, \eps} \mathcal{R}_2\;.
\end{equation}
\end{definition}

 With this definition we have the following.

\begin{definition}\label{def:sim}
Let $\eps=\eps(\lambda)\in[0,1]$ be a function of a security parameter $\lambda\in\N$. We say that a protocol $\pi=(\pi_A,\pi_B)$ constructs
a resource $\mathcal{S}$ from a resource $\mathcal{R}$ with (computational) error (or distance) $\eps$ if:
\begin{itemize}
\item \textbf{Correctness:} $\pi_{AB} \mathcal{R} \approx_{c, \eps} \mathcal{S}$.
\item \textbf{Security:} There exists an efficient simulator $\sigma$ such that $\pi_A \mathcal{R} \approx_{c, \eps} \mathcal{S} \sigma$.
\end{itemize}
(Here $\pi_{AB}$, $R$, $S$ and $\sigma$ may all implicitly depend on $\lambda$.)
\end{definition}

The first condition expresses the fact that if Alice and Bob follow the instructions of the protocol, the resulting resource behaves as the ideal one. The second condition expresses the fact that if Bob does not follow the protocol, any attack he performs on the real protocol can be mapped to an attack on the ideal protocol. This mapping is referred to as a simulator. Note that Definition~\ref{def:sim} implicitly assumes that Alice always behaves honestly; this need not be the case in general but always holds in the context of this paper. 

\subsection{Rigidity}

\begin{definition}
Let finite-dimensional Hilbert spaces $\mH_{\reg{A}}$ and $\mH_{\reg{A}'}$ and operators $R \in\Lin(\mH_{\reg{A}})$ and $S\in\Lin(\mH_{\reg{A}'})$ be functions of a parameter $\delta>0$ (the dependence on $\delta$ is left implicit in the notation). We say that $R$ and $S$ are $\delta$-isometric with respect to $\ket{\psi} \in \mH_{\reg{A}} \otimes \mH_{\reg{B}}$, and write $R\simeq_\delta S$, if there exists an isometry $V:\mH_{\reg{A}}\to\mH_{\reg{A}'}$ such that 
$$\big\|( R-V^\dagger SV)\otimes \Id_{\reg{B}} \ket{\psi}\big\|^2=O(\delta).$$
We sometimes write the isometry as a CPTP map $\Phi(R) = V R V^\dagger$ for $R\in \Lin(\mH_A)$, and also write $\Phi(\ket{\phi})$ for $V\ket{\phi}$, $\Phi(\sigma)$ for $V\sigma V^\dagger$. 
If $V$ is the identity, then we further say that $R$ and $S$ are $\delta$-equivalent, and write $R\approx_\delta S$ for $\| ( R- S) \otimes \Id_{\reg{B}} \ket{\psi}\|^2=O(\delta)$.
\end{definition}

The following can be shown by a standard application of Jordan's lemma. Furthermore, the isometry $V$ can be implemented using the ``swap'' isometry as in~\cite{mckague2012robust}. 

\begin{lemma}\label{lem:pauli-c}
Let $\ket{\psi} \in \mH_{\reg{A}} \otimes \mH_{\reg{B}}$ and $Z,X,X'$ observables on $\mH_{\reg{A}}$ such that $\{Z,X\}\approx_\delta 0$ and $\{Z,X'\}\approx_\delta 0$. 
Then there exist $\delta' = O(\sqrt{\delta})$, an isometry $V:\mH_\reg{A}\to \C^2\otimes \mH_{{\reg{A}'}}$, and Hermitian commuting $A_X,A_Y$ on $\mH_{{\reg{A}'}}$ such that $A_X^2+A_Y^2=\Id$ and 
\[Z \simeq_{\delta'} \sigma_Z \otimes \Id,\quad X \simeq_{\delta'} \sigma_X \otimes \Id,\quad\text{and}\quad X' \simeq_{\delta'} \sigma_X \otimes A_X + \sigma_Y \otimes A_Y\;.\]
Furthermore, there exists a polynomial-time algorithm that given explicit circuits implementing $Z,X$ and $X'$ as input returns an explicit circuit that implements the isometry $V$. 
\end{lemma}

\subsection{Delegated quantum computation} \label{subsec:dqc}

\subsubsection{Ideal functionalities}

We recall the ideal resources for blind and verifiable delegated quantum computation (DQC), as defined in~\cite{dunjko2013composable}. We start with blindness.

\begin{definition}[Definition 4.1 in~\cite{dunjko2013composable}]
The ideal DQC resource $\mS^{blind}$ which provides both correctness and blindness takes an input $\psi_A$ at Alice's interface, but no honest input at Bob's interface. Bob's filtered interface has a control bit $b$, set by default to $0$, which he can flip to activate the other filtered functionalities. The resource $\mS^{blind}$ then outputs the permitted leak $\ell^{\psi_A}$ at Bob's interface, and accepts two further inputs, a state $\psi_B$ and a map description $\proj{\mE}$. If $b=0$, it outputs the correct result $\mU(\psi_A)$ at Alice's interface; otherwise it outputs Bob's choice $\mE(\psi_{AB})$. 
\end{definition}

Next we give the definition for blindness and verifiability. The main difference is that the ideal functionality is no longer allowed to return an output of Bob's choice at Alice's interface. 

\begin{definition}[Definition 4.2 in~\cite{dunjko2013composable}]
The ideal DQC resource $\mS^{blind}_{verif}$ which provides correctness, blindness and verifiability takes an input $\psi_A$ at Alice's interface, and two filtered control bits $b$ and $c$ (set by default to $0$). If $b=0$, it outputs the correct result $\mU(\psi_A)$ at Alice's interface. If $b=1$, it outputs the permitted leak $\ell^{\psi_A}$ at Bob's interface, then reads the bit $c$, and conditioned on its value, it either outputs $\mU(\psi_A)$ or $\ket{ERR}$ at Alice's interface. 
\end{definition}

We provide the definitions of the ideal resources for two variants of random remote state preparation introduced in \cite{dunjko2016blind}:

\begin{definition}[Definition 11 in~\cite{dunjko2016blind}]
The ideal resource called the strong random remote state preparation $(\RSPS)$ has two interfaces A, and B, standing for Alice and Bob.
The resource first selects an angle $\theta$ (from the set of 8 states) chosen uniformly at random. Bob's interface has a filtered functionality comprising a bit $c$ which Bob can pre-set to zero or one, depending on whether he will behave maliciously. If Bob pre-sets $c=0$, the resource outputs the state $\proj{+_{\theta}}$ on Bob's interface.
If Bob pre-sets $c=1$, it awaits a description of a CPTP map $\mathcal{E}$ from Bob.  
Once the set is received, the functionality outputs $\mathcal{E}(\proj{+_{\theta}})$ at Bob's interface. In both cases, the resource outputs the angle $\theta$ at Alice's interface.
\end{definition}

\begin{definition}[Definition 8 in~\cite{dunjko2016blind}]
The ideal resource called the \emph{random remote blind state preparation for blindness} $\RSPB$ has two interfaces A, and B, standing for Alice and Bob.
The resource first selects a $\theta$ chosen uniformly at random. Bob's interface has a filtered functionality comprising a bit $c$ which Bob can pre-set to zero or one, depending on whether he will behave maliciously. If Bob pre-sets $c=0$, the resource outputs the state $\proj{+_{\theta}}$ on Bob's interface.
If Bob pre-sets $c=1$, it awaits the set $\{(\theta, \left[ \rho^{\theta}\right] )\}_{\theta}$ from Bob, where $\left[ \rho^{\theta}\right]$ denotes the classical description of a quantum state, with the property that $\rho^{\theta} + \rho^{\theta+\pi}  =\rho^{\theta'} + \rho^{\theta'+\pi}, \forall \theta,\theta'.$ 
If the states Bob inputs do not satisfy the property above, the ideal functionality ignores the set Bob has input and awaits a new valid set.
Once the set is received, the functionality outputs $\rho^{\theta}$ at Bob's interface. In both cases, the resource outputs the angle $\theta$ at Alice's interface.
\end{definition}

Briefly, the difference between \RSPS\ and \RSPB\ is as follows. In \RSPS\ if Alice accepts she receives a random angle $\theta$, and Bob receives $\mathcal{E}(\ket{+_{\theta}})$, for some CPTP map $\mathcal{E}$, that is \emph{independent of $\theta$}. In the weaker \RSPB, the possible states, $\{ \rho_{\theta} \}_{\theta \in \Theta}$, that Bob receives should satisfy
\begin{equation} \label{eq:rho_consistency}
\rho_{\theta} + \rho_{\pi + \theta} = \Id\;.
\end{equation}
Importantly, there is no requirement that Bob's state is the correct $\ket{+_{\theta}}$ state, up to to the action of an independent CPTP map. It is precisely the possibility that the deviation map can depend on $\theta$ that makes \RSPB\ unsuitable for verifiability (though it does provide blindness).

\subsubsection{Local criteria}

Dunjko et al.~\cite{dunjko2013composable} give ``local'' criteria, $\delta$-local-blindness and independent $\delta$-local-verifiability, that can be used to establish the security of a protocol for delegated quantum computation in the AC framework. Their definitions are geared to showing information-theoretic security. We adapt them to the setting of computational security, as follows. 

\begin{definition}\label{def:local-blind}
A DQC protocol provides $\delta$-local-blindness if for all efficient adversaries $\{\mF_i:\Lin(\mH_{CB})\to\Lin(\mH_{CB})\}$ in the protocol there is an efficient CPTP map $\mF:\Lin(\mH_B)\to\Lin(\mH_B)$ such that for all efficiently preparable $\psi_{ABR}$, 
\begin{equation}\label{eq:local-blind-0}
\Tr_A \circ \mP_{AB} (\psi_{ABR}) \approx_{c,\delta} \mF \circ \Tr_A(\psi_{ABR})\;,
\end{equation}
where $\mP_{AB}$ is the map corresponding to an execution of the protocol with an honest Alice and a Bob specified by the maps $\mF_i$ and $\circ$ denotes composition. When no map acts on a space, it is to be assumed that the identity is applied. 
\end{definition}

\begin{definition}\label{def:local-ver}
A DQC protocol provides independent $\delta$-local verifiability if for all efficient adversaries $\{\mF_i:\Lin(\mH_{CB})\to\Lin(\mH_{CB})\}$ in the protocol there exist efficient alternative maps $\{\mF'_i:\Lin(\mH_{CBB'})\to\Lin(\mH_{CBB'})\}$ such that the following hold:
\begin{enumerate}
\item For all efficient initial states $\psi_{AR_1}\otimes \psi_{R_2B}$ there is a $0\leq p^\psi \leq 1$ such that 
\begin{equation}\label{eq:local-ver-0}
 \rho_{AR_1}^\psi \approx_{c,\delta} p^\psi ( \mU \otimes \Id_{R_1})(\psi_{AR_1})+(1-p^\psi)\proj{ERR}\otimes \psi_{R_1}\;,
\end{equation}
where $\rho_{AR_1}^\psi$ is the final state of Alice and the first part of the reference system;
\item For all efficient initial states $\psi_{ABR}$,
\begin{equation}\label{eq:local-ver-1}
 \Tr_A \circ Q_{AB'} \circ \mP_{AB}(\psi_{ABR}) \approx_{c,\delta} \Tr_A \circ \mP'_{ABB'}\;,
\end{equation}
where $\mP_{AB}$ and $\mP'_{ABB'}$ are the maps corresponding to an execution of the protocol with an honest Alice and a Bob specified by the maps $\mF_i$ and $\mF_i'$ respectively, and $Q_{AB'}:\Lin(\mH_A)\to\Lin(\mH_{AB'})$ is a map which generates from $A$ a system $B'$ that contains a copy of the information whether Alice accepts or rejects.
\end{enumerate}
\end{definition}

The following is an analogue of~\cite[Corollary 6.9]{dunjko2013composable} for the computational setting. 

\begin{theorem}\label{thm:comp-blind}
If a DQC protocol $\pi$ implementing a unitary transformation $\mU$ is $\delta_c$-correct and provides $\delta_b$-local-blindness and independent $\delta_v$-local-verifiability for all efficient inputs that are classical on $A$,\footnote{We make the restriction that the input is classical for convenience; a more general version of the theorem, with some loss in parameters, applies to quantum inputs. See~\cite[Corollary 6.9]{dunjko2013composable} for details.} for some $\delta_c,\delta_b,\delta_v\geq 0$,  then it constructs $\mS_{verif}^{blind}$ computationally within $\eps = \max(\delta_c,2\delta_b+4\sqrt{\delta_v})$.
\end{theorem}

\begin{proof}
The proof is identical to the proof of~\cite[Corollary 6.9]{dunjko2013composable} except for ensuring that the simulator is computationally efficient. The first step is to combine local-blindness and local-verifiability to obtain the condition of local-blind-verifiability, i.e. the existence of maps $\mF^{ok}$ and $\mF^{ver}$ such that, using the notation from Definition~\ref{def:local-ver}, 
\begin{equation}\label{eq:comp-blind-1}
\rho_{AR_1R_2B}^\psi \approx_{c,\delta} (\mU\otimes \Id_{R_1R_2} \otimes \mF^{ok})(\psi_{AR_1}\otimes \psi_{R_2B}) + \proj{ERR} \otimes \psi_{R_1} \otimes (\Id_{R_2} \otimes \mF^{err})(\psi_{R_2B})\;.
\end{equation}
The maps $\mF^{ok}$ and $\mF^{ver}$ can be defined from $\mF'$ as in the proof of~\cite[Lemma 6.6]{dunjko2013composable} as
\[ \mF^{ok} = \Tr_{B'} \circ \mP_{B'}^{ok} \circ \mF'\;,\qquad \mF^{err} = \Tr_{B'} \circ \mP_{B'}^{err} \circ \mF'\;,\]
where $\mP_{B'}^{ok}$ and $\mP_{B'}^{err}$ are the projection on the corresponding states of $B'$. Clearly these maps can be implemented efficiently, given that $\mF'$ can. Next we need to show that any efficient distinguisher $\mD$ for~\eqref{eq:comp-blind-1} contradicts either local-blindness or local-verifiability. Write
\[ \rho_{AR_1R_2B}^\psi \,=\, \phi_{ARB}^{ok} + \proj{ERR} \otimes \phi_{RB}^{err}\;.\]
If $\mD$ is a distinguisher for~\eqref{eq:comp-blind-1} with advantage $\delta$, by Lemma~\ref{lem:dist} there exist distinguishers $D^{ok}$ between $\phi_{ARB}^{ok}$ and $(\mU\otimes \Id_{R_1R_2} \otimes \mF^{ok})(\psi_{AR_1}\otimes \psi_{R_2B}) $, and $D^{err}$ between $ \phi_{RB}^{err}$ and  $\psi_{R_1} \otimes (\Id_{R_2} \otimes \mF^{err})(\psi_{R_2B})$, with advantage $\delta_1$ and $\delta_2$ respectively such that $\delta_1+\delta_2 \geq \delta$.

Consider first the case of $\phi_{RB}^{err}$. Using~\eqref{eq:local-blind-0},~\eqref{eq:local-ver-1} and Lemma~\ref{lem:dist} it follows that $\delta_2 \leq \delta_b+ \delta_v$. Consider next $\phi_{ARB}^{ok}$. Using the triangle inequality,  $\mD$ must distinguish between $\phi^{ok}_{AR_1R_2B}$ and $p^\psi \mU(\psi_{AR_1} )\otimes \phi^{ok}_{R_2 B}$, and between $p^\psi \mU(\psi_{AR_1} )\otimes \phi^{ok}_{R_2 B}$ and $\mU(\psi_{AR_1})\otimes \mF_B^{ok}(\psi_{R_2B})$, with advantage $\delta'_1$ and $\delta'_2$ respectively such that $\delta'_1 + \delta'_2 \geq \delta_1$. 

Using property~\eqref{eq:local-ver-0} of local verifiability and Lemma~\ref{lem:dist} it follows that no efficient distinguisher can distinguish $\phi_{AR_1}^{ok}$ from $p^\psi \mU(\psi_{AR_1})$ with advantage larger than $\delta_v$. Using that the state $\mU(\psi_{AR_1})$ is efficiently preparable, a specific distinguisher would be to perform a swap test with that state. It follows that $\Tr(\mU(\psi_{AR_1})\phi_{AR_1}^{ok}) \leq 2\delta_v $. Using the relation between fidelity and trace distance and Uhlmann's theorem it follows that $\|\phi_{AR_1R_2B}^{ok} - p^\psi \mU(\psi_{AR_1}) \otimes \phi_{R_2 B}^{ok} \|_1 \leq \sqrt{4\delta_v}$, so by Lemma~\ref{lem:comp-inf} it holds that $\delta'_1 \leq \sqrt{4\delta_v}$ as well. 

Finally, using again~\eqref{eq:local-blind-0},~\eqref{eq:local-ver-1} and Lemma~\ref{lem:dist} it follows that $\delta'_2 \leq \delta_b+\delta_v$. 

Having established~\eqref{eq:comp-blind-1} for some $\delta \leq 2\delta_b+4\sqrt{\delta_v}$, it remains to show the existence of a simulator $\sigma_B$ such that $\pi \approx_{c,\delta} \mS_{verif}^{blind}\sigma_B$. The simulator is identical to the simulator constructed in the proof of~\cite[Theorem 5.2]{dunjko2013composable}: the simulator simply interacts with Bob as Alice would in the protocol $\pi$, using an arbitrary input $\psi_B$ instead of Alice's real input $\psi_A$. This simulator is clearly efficient. The remainder of the argument is exactly the same, and we omit the details. 
\end{proof}

%% file: protocol.tex
In this section we describe and analyze our implementation of the ideal resource for random remote state preparation with verification, \RSPV. 
We refer to this implementation as the real protocol (to be contrasted with the ideal protocol/functionality, introduced in the next section).
As mentioned in the introduction, our implementation builds upon the randomness certification protocol from~\cite{brakerski2018cryptographic} by, informally, performing the protocol modulo $8$ (instead of modulo $2$) and adding tests inspired from the study of quantum random access codes to verify that the prover prepares the right state, up to a local isometry. 

We start by recalling the definition of a QRAC in Section~\ref{sec:qrac}, and show a new result about rigidity of $2\mapsto 1$ QRAC. In Section~\ref{sec:qubit-test} we introduce the main building block for our protocol; we analyze its soundness and rigidity properties in Section~\ref{sec:qubit-rigidity}. Finally, we describe and analyze our protocol for \RSPV\ in Section~\ref{sec:protocol}.

\subsection{Quantum random access codes}
\label{sec:qrac}

\begin{definition}
A $2\mapsto 1$ quantum random access code (QRAC) is specified by four single-qubit density matrices $\{\phi_u\}_{u\in \{1,3,5,7\}}$ and two single-qubit observables $X_0$ and $X_2$. For $u\in \{1,3,5,7\}$ let $u_0,u_2\in\{0,1\}$ be such that $u_0 = 0$ if and only if $u\in\{1,7\}$ and $u_2 = 0$ if and only if $u\in\{1,3\}$.\footnote{The motivation for the somewhat obscure indexing scheme will become clear later.} The success probability of the QRAC is defined as  
\[ \frac{1}{4} \sum_{u\in\{1,3,5,7\}} \frac{1}{2} \sum_{i\in\{0,2\}}\Tr\big(X_{i}^{u_i} \phi_u\big) \;.\]
\end{definition}

Let $\opt_Q = \frac{1}{2} + \frac{1}{2\sqrt{2}}$. As shown in~\cite[Theorem 3]{ambainis2008quantum}, the highest possible success probability of a single-qubit $2\mapsto 1 $ QRAC is $\opt_Q$. More generally, we have the following rigidity statement. 

\begin{lemma}\label{lem:qrac}
Let $\{\phi_u\}$ and $X_0,X_2$ be a $2\mapsto 1$ QRAC whose success probability is at least $(1-\delta)\opt_Q$, for some $0\leq \delta < 1$.
Then 
\[ \frac{1}{4}\sum_{u\in\{1,3,5,7\}} \Tr\big( \{X_0,X_{2}\}^2 \phi_u \big) \,=\,O({\delta})\;.\]
\end{lemma}

\begin{proof}
Assume without loss of generality that both observables $X_0$ and $X_2$ are in the plane specified by $\sigma_X$ and $\sigma_Y$. Let $v_0 = (x_0,y_0,0)$ and $v_2 = (x_1,y_1,0)$ be the Bloch sphere representation of the eigenvalue-$1$ eigenvector of $X_0$ and $X_2$ respectively, i.e. real unit vectors such that for $i\in\{0,2\}$, $X_i = x_i\sigma_X+y_i\sigma_Y$. As shown in~\cite[Section 3.4]{ambainis2008quantum}, the optimal success probability of any $2\mapsto 1$ QRAC based on $X_0$ and $X_2$ is $\frac{1}{2}(1+\frac{S}{8})$, where $S = 2\|v_0+v_2\|+2\|v_0-v_2\|$. In order for the QRAC to achieve a success probability of $(1-\delta)\opt_Q$ it is necessary that $\|v_0+v_2\|+\|v_0-v_2\| \geq 2\sqrt{2}-16\delta$. Using $4=\|v_0+v_2\|^2+\|v_0-v_2\|^2$ it follows that $|\|v_0+v_2\|^2 - \|v_0-v_2\|^2| = O(\sqrt{\delta})$, thus $|v_0\cdot v_2|=O(\sqrt{\delta})$. Since $\{X_0,X_2\} = 2iv_0\cdot v_2 \Id$, we obtain $\|\{X_0,X_2\}\|^2=O(\delta)$. 
\end{proof}

The following simple test will be used later to estimate the success probability of a QRAC. We introduce it here to set some notation. 

\begin{definition}\label{def:qrac-test}
Let $\{\phi_u\}_{u\in\{0,1,2,\ldots,7\}}$ be arbitrary density matrices. 
In the \emph{QRAC test}, the prover is given $\phi_u$ for a uniformly random $u\in\{1,3,5,7\}$. The verifier sends a uniformly random $\theta\in\{0,2\}$ to the prover, who replies with a bit $v$. If $v \neq u_\theta$, the verifier sets $flag\leftarrow fail_Q$.
\end{definition}

\subsection{The qubit preparation test}
\label{sec:qubit-test}

The \emph{qubit preparation test} described in Figure~\ref{fig:qubit} forms the main building block of our remote state preparation protocol. The test relies on an extended variant of the family of claw-free functions used in~\cite{brakerski2018cryptographic}, introduced in~\cite{mahadev2018classical} and called an \emph{extended noisy trapdoor claw-free family} (ENTCF). We recall the definition of an ENTCF family $(\mF,\mG)$ in Appendix~\ref{sec:entcf}, where we also present the main properties needed. For the purposes of this section it is sufficient to think of both $\mF$ and $\mG$ as families of pairs of functions, $(f_{k,0},f_{k,1})$ or $(g_{k,0},g_{k,1})$, where $k$ denotes a public key, such that both functions in an $f$-pair (also called claw-free pair) are bijections with the same domain and range, while both functions in a $g$-pair (also called injective pair) are bijections with the same domain but non-intersecting ranges, and such that moreover given a key $k$ it is computationally impossible to distinguish if $k$ corresponds to a claw-free or an injective pair.

\begin{figure}[htbp]
\rule[1ex]{16.5cm}{0.5pt}\\
Let $\lambda$ be a security parameter.
\begin{enumerate}
\item The verifier selects $G\leftarrow_U \{0,1\}$. If $G=0$ they sample a key $(k,t_k)\leftarrow \Gen_\mF(1^\lambda)$. If $G=1$ they sample $(k,t_k)\leftarrow \Gen_\mG(1^\lambda)$. The verifier sends $k$ to the prover and keeps the trapdoor information $t_k$ private. 
\item The prover returns a  $y \in \mY$ to the verifier. If $G=0$, for $b\in\{0,1\}$ the verifier uses the trapdoor to compute $\hat{x}_b\leftarrow \Inv_\mF(t_k,b,y)$. If $G=1$, the verifier computes $(\hat{b},\hat{x}_{\hat{b}})\leftarrow \Inv_\mG(t_k,y)$.
\item The verifier performs either of the following with equal probability.
\begin{enumerate}
\item (\emph{preimage test}) The verifier requests a preimage. The prover returns 
$(b,x)\in\{0,1\}\times \mX$. If $G=0$ and $x\neq \hat{x}_b$, or if $G=1$ and $(b,x)\neq(\hat{b},\hat{x}_{\hat{b}})$, the verifier sets $flag\leftarrow fail_p$.
\item (\emph{measurement test}) The verifier requests an equation $d\in \Z_8^w$ from the prover. If $G=0$, the verifier computes $\hat{\theta}=\hat{\theta}(d)$ and $\hat{v}=\hat{v}(d)$. 
\item[] The verifier performs either of the following tests with equal probability:
\begin{enumerate}
\item[(i)] (\emph{$Z$-measurement test}) The verifier sends the label $Z$ to the prover. The prover replies with a bit $b$. If $G=1$ and $b\neq \hat{b}$, the verifier sets $flag\leftarrow fail_Z$. 
\item[(ii)] (\emph{$X_\theta$-measurement test}) The verifier selects $\theta\leftarrow_U\{0,1,2,3\}$ and sends $\theta$ to the prover. The prover responds with a bit $v$. If $G=0$, depending on the value of $\theta$, the verifier performs one the following tests:
\begin{enumerate}
\item[A.] If $\theta=\hat{\theta}$ but $v\neq \hat{v}$, the verifier sets $flag\leftarrow fail_X$. 
\item[B.] If $\theta\in\{0,2\}$ and $\hat{\theta}\in \{1,3\}$  the verifier performs the QRAC test (Definition~\ref{def:qrac-test}). 
\end{enumerate}
\end{enumerate}
\end{enumerate}
\end{enumerate}
\rule[1ex]{16.5cm}{0.5pt}
\caption{The qubit preparation test.}
\label{fig:qubit}
\end{figure}

We first show a completeness property of the qubit preparation test. 

\begin{lemma}[Completeness]\label{lem:qubit-completeness}
There is an efficient quantum prover that is accepted with probability negligibly close to $1$ in the security parameter $\lambda$ in each of the preimage test and part A. of the $X_\theta$-measurement test, and with probability negligibly close to $\opt_Q$ in part B. of the $X_\theta$-measurement test (Figure~\ref{fig:qubit}). Moreover, in case the verifier selects $G=0$ and a key $k$, after having returned a $y\in \mY$ in step 2. and an equation $d\in\Z_8^w$ at the beginning of step 3. the state of the prover is the state 
\begin{equation}\label{eq:state-2}
 \frac{1}{\sqrt{2}} \big(  e^{\frac{2i\pi}{8} d\cdot J(x_0)}\ket{0} \ket{x_0} + e^{\frac{2i\pi}{8} d\cdot J(x_1)}\ket{1}\big)\;,
\end{equation}
where $x_0,x_1$ are the two preimages of $y$ under $f_{k,0}$ and $f_{k,1}$ respectively and $J$ is a simple map described in Appendix~\ref{sec:entcf}. 
\end{lemma}

\begin{proof}
The honest strategy for the prover is as follows. Upon receipt of a key $k$ that specifies a pair of functions $f_{k,0}$ and $f_{k,1}$ the prover prepares a state $\ket{+} = \frac{1}{\sqrt{2}}\ket{0}+\frac{1}{\sqrt{2}}\ket{1}$, adjoins a uniform superposition over all $x\in \mX$, evaluates $f$ in superposition, and measures the outcome $y$. The result is the state 
\begin{equation}\label{eq:state-1}
 \frac{1}{\sqrt{2}} \big( \ket{0} \ket{x_0} + \ket{1}\ket{x_1}\big)\;,
\end{equation}
where $x_0$ and $x_1$ are the unique preimages of $y$ under $f_{k,0}$ and $f_{k,1}$ respectively.\footnote{Here for clarity we ignore the fact that $f_{k,b}$ ranges over the set of \emph{distributions} over $\mY$, rather than over $\mY$ itself. For details on how the prover can construct the state~\eqref{eq:state-1} with success probability exponentially close to $1$ in $\lambda$, we refer to~\cite{brakerski2018certifiable}.} (In case $G=1$ the state further collapses to a single $\ket{b,x_b}$.)

If the verifier requests a preimage, the prover measures in the computational basis and returns $(b,x_b)$. If the verifier requests an equation, the prover first evaluates the map $J$ on the second register and then measures all but the first register in the Fourier (over $\Z_8$) basis to obtain a string $d\in\Z_8^w$. The resulting state is 
\begin{equation}\label{eq:state-3}
 \frac{1}{\sqrt{2}} \big(  e^{\frac{2i\pi}{8} d\cdot J(x_0)}\ket{0} \ket{x_0} + e^{\frac{2i\pi}{8} d\cdot J(x_1)}\ket{1}\big)\;,
\end{equation}
where the inner products are taken modulo $8$. Finally, the prover measures the qubit in~\eqref{eq:state-2} in the requested basis, $\sigma_Z$ in the case of a $Z$-measurement or $\sigma_{X,\theta\frac{\pi}{4}}$ in the case of an $X_\theta$-measurement, to produce its answer. 
\end{proof}

\subsection{Rigidity}
\label{sec:qubit-rigidity}

In this section we show that any prover, or \emph{device}, that succeeds with probability close to optimum in the qubit preparation test (Figure~\ref{fig:qubit}) must perform measurements that obey a form of rigidity. We generally use $\eps$ to denote the failure probability of the device in the test or one of its parts (the definition of $\eps$ will always be specified in context), and always assume that $\eps$ is larger than any term of the form $\negl(\lambda)$. The main result of the section is the following. 

\begin{lemma} \label{lem:rigidity}
Let $\eps>0$ and $\lambda$ a security parameter assumed to be chosen large enough so that $\eps =\omega( \negl(\lambda))$. Suppose that a quantum polynomial-time prover succeeds in the qubit preparation test with probability at least $1-\eps$. Let $Z$ be the observable associated with the prover's strategy in step (b)(i) of the protocol, and $\{X_\theta\}_{\theta\in \{0,1,2,3\}}$  the observables associated with the prover's strategy in step (b)(ii). 

Then there exists a universal constant $c > 0$, a $\delta = O(\eps^c)$,  an efficiently computable isometry $\Phi:\mH_B \to \C^2\otimes \mH_{B'}$, where $\mH_B$ is the Hilbert space on which the prover's observables act,  and a state $\ket{\aux}\in \mH_{B'}\otimes  \mH_{B''}$, where $mH_{B'}$ is a purifying system for Bob's initial state in $\mH_{B}$, such that under the isometry $\Phi$ the following hold:
\begin{itemize}
\item In case $G=1$, the joint state of the bit $b$ and the prover's post-measurement state in step (b) of the protocol, after having returned an equation $d$, is computationally indistinguishable from a state that is within $\delta$ trace distance of  
\[\sum_b \proj{b} \otimes \proj{b} \otimes \proj{\aux} \;.\]
\item In case $G=0$, the joint state of the angle $\hat{\theta}$, the bit $\hat{v}$ and the prover's post-measurement state in step (b) of the protocol, after having returned an equation $d$,  is computationally indistinguishable from a state that is within $\delta$ trace distance of  
\[\sum_{\theta \in \{0, 1, 2, 3\} ,v \in \{0, 1\}} \proj{\theta}\otimes \proj{v} \otimes \proj{+_{\theta\frac{\pi}{4}+v\pi}} \otimes \proj{\aux} \;.\]
\end{itemize}
\end{lemma}
The proof of Lemma~\ref{lem:rigidity} is given at the end of Section~\ref{sec:xpartb}. 
We start by introducing notation to model the behavior of an arbitrary prover in the test. 

\subsubsection{Devices}

\begin{definition}\label{def:device}
A device $D = (\phi,\Pi,M,Z,\{X_\theta\}_{\theta\in\{0,1,2,3\}})$ is specified by the following. 
\begin{enumerate}
\item A (not necessarily normalized) positive semidefinite $\phi \in \Pos(\mH_\reg{D}\otimes \mH_\reg{Y})$. Here $\mH_\reg{D}$ is an arbitrary space private to the device, and $\mH_{\reg{Y}}$ is a space of the same dimension as the cardinality of the set $\mY$, also private to the device. (We think of $\phi$ as the state of the device immediately prior to returning the commitment string $y$. In particular, $\phi$ implicitly depends on the key $k\in \mK_\mF\cup \mK_\mG$.)
 For every $y\in\mY$, define
$$\phi_y \,=\, (\Id_{\reg{D}} \otimes \bra{y}_\reg{Y})\,\phi\,(\Id_{\reg{D}} \otimes \ket{y}_\reg{Y})\,\in\,\Pos(\mH_\reg{D})\;.$$
Note that $\phi_y$ is not normalized, and $\sum_{y\in\mY} \Tr(\phi_y)=\Tr(\phi)$. 
\item For every $y\in\mY$,
\begin{enumerate}
\item A projective measurement $\{\Pi_y^{(b,x)}\}$ on $\mH_\reg{D}$, with outcomes $(b,x)\in \{0,1\}\times\mX$. For each $y$, this measurement has two designated outcomes $(0,x_0)$ and $(1,x_1)$, which are the answers that are accepted in the preimage test; recall that we use the notation $V_{y}$ for this set. For $b\in\{0,1\}$ we use the shorthand $\Pi_y^b = \Pi_y^{(b,x_b)}$, $\Pi_y= \Pi_y^0+\Pi_y^1$, and $\Pi_y^2 = \Id - \Pi_y^0-\Pi_y^1$.
\item A projective measurement $\{M_y^{d}\}$ on $\mH_\reg{D}$, with outcomes $d\in \Z_8^w$. 
\item A binary observable $Z$ on  $\mH_\reg{D}$.
\item For every $\theta \in \{0,1,2,3\}$, a binary observable $X_\theta$ on  $\mH_\reg{D}$.
\end{enumerate}
\end{enumerate}
\end{definition}

By Naimark's theorem, up to increasing the dimension of $\mH_\reg{D}$ the assumption that $\{\Pi_y^{(b,x)}\}$, $\{M_y^{d}\}$ and $Z$, $X_\theta$ are projective is without loss of generality. For notational convenience we often drop the subscript $y$ from the measurements $\Pi_y$ and $M_y$, and the state $\phi_y$.  

\begin{definition}[Efficient devices]
We say that a device $D = (\phi,\Pi,M,Z,\{X_\theta\})$ is \emph{efficient} if the state $\phi$ can be prepared efficiently, and each of the measurements can be implemented efficiently. 
\end{definition}



We introduce notation for some post-measurement states of a device.

\begin{definition}\label{def:phitheta}
Let $D = (\phi,\Pi,M,Z,\{X_\theta\})$ be a device. 
Let $\theta\in\{0,1,2,3\}$ and $v\in\{0,1\}$. Define a sub-normalized density matrix
\begin{equation}
\phi_{y,\theta,v} \,=\, \sum_{d:\, (\hat{\theta}(d),\hat{v}(d))= (\theta,v)} \big( \Id_\reg{Y}\otimes M_y^d \big)\,\phi_y\, \big(\Id_\reg{Y} \otimes M_y^d\big)\;.
\end{equation}
We sometimes omit $y$ and write $\phi_{\theta,v}$ for the same state. Note that since we assumed that $\{M_y^d\}$ is projective, the $8$ states $\phi_{\theta,v}$ are orthogonal. We write $\phi_{\theta} = \phi_{\theta,0}+\phi_{\theta,1}$.
\end{definition}

\subsubsection{Preimage test}

In this section we draw consequences from the assumption that a device succeeds with probability at least $1-\eps$ in the preimage test. 

\begin{lemma}\label{lem:preimage}
Let $D=(\phi,\Pi,M,Z,\{X_\theta\})$ be an efficient device that succeeds with probability at least $1-\eps$ in the preimage test, for some $0\leq \eps \leq 1$.
Then there is an efficient device $D'=(\phi',\Pi,M,Z,\{X_\theta\})$ such that $\|\phi'-\phi\|_1 = O(\sqrt{\eps})$ and such that $D'$ succeeds with probability negligibly (in the security parameter $\lambda$) close to $1$ in the preimage test. In particular, for any $k\in \mK_\mF$ the state of $D'$ after having returned $y$ has the form 
\begin{equation}\label{eq:preimage-phi}
 \ket{\phi'_y} = \sum_{b\in\{0,1\}} \ket{b,x_b} \ket{\phi_{y,b}}\;,
\end{equation}
where for $b\in\{0,1\}$, $x_b= \Inv_\mF(t_k,b,y)$, $\ket{\phi_{y,0}}$ and $\ket{\phi_{y,1}}$ are arbitrary, and the basis is chosen such that the measurement $\{\Pi^{(b,x_b)}\}$ is a computational basis measurement of the first two registers. Similarly, for $k\in\mK_\mG$ the same state can be expressed as
\begin{equation}\label{eq:preimage-phi-g}
 \ket{\phi'_y} =  \ket{\hat{b},x_{\hat{b}}} \ket{\phi_{y,\hat{b}}}\;,
\end{equation}   
where $(\hat{b},x_{\hat{b}})= \Inv_\mG(t_k,y)$.
\end{lemma}

\begin{proof} The proof is analogous to the reduction to a ``perfect prover'' shown in~\cite[Claim 7.2]{mahadev2018classical_arxiv}, and we only sketch it here. 
Given $\phi_y$, the device can evaluate CHK$_{\mathcal{F}}$  in superposition to check if it would succeed in the preimage test. The device $D'$ then repeatedly prepares $\phi$ and measures $y$ as $D$ would, until it has obtained a state $\phi_y$ that passes the preimage test with certainty (or until a polynomial number of attempts to do so have failed). The distance between $D$ and $D'$ is bounded by the gentle measurement lemma (Lemma~9 in~\cite{winter1999coding}).  
\end{proof}

\begin{lemma}\label{lem:d-uniform}
Let $D$ be an efficient device that succeeds with probability $1$ in the preimage test.
Then for every $\theta\in\{0,1,2,3\}$ and $v\in\{0,1\}$, no polynomial-time quantum procedure can predict $\hat{\theta}(d)$ given $(y,d,\phi_{\theta,v})$ with advantage non-negligibly larger than $\frac{1}{4}$. Moreover, for every $\theta\in\{0,1,2,3\}$ no polynomial-time quantum procedure can predict $\hat{v}(d)$ given $(y,d,\phi_{\theta,v},\theta)$ with advantage non-negligibly larger than $\frac{1}{2}$. 

In particular, 
the joint distribution of $(\hat{\theta}(d),\hat{v}(d))$ computed by the verifier in the measurement test is negligibly close to uniform, where the probability is taken over the device's actions, including the choice of $y$ and $d$. 
\end{lemma}

\begin{proof}
Suppose for contradiction that there exists a distinguisher that achieves success probability noticeably larger than $\frac{1}{8}$, where the probability is over $y$ and $d$ as computed by the device as well as the distinguisher's internal randomness. Suppose first that the distinguisher can predict $\hat{\theta}(d))$ with advantage noticeably larger than $\frac{1}{4}$.
Using the collapsing property (Lemma~\ref{lem:collapse}) and the fact that $\{M_y^d\}$ is efficient and the distinguisher are assumed efficient, it is still the case that the distinguisher has advantage noticeably larger than $\frac{1}{4}$ in predicting $\hat{\theta}(d))$ when the device first measures $\{\Pi_y^{(b,x_b)}\}$ to obtain $(b,x_b)$ and then only applies $\{M_y^d\}$ to obtain $d$. This contradicts the hardcore bit property~\eqref{eq:adaptive-hardcore-2}.

Similarly, if the distinguisher has advantage noticeably larger than $\frac{1}{2}$ in predicting $\hat{v}(d)$, conditioned on its guess for $\hat{\theta}(d)$ being correct, using the collapsing property we construct an adversary that contradicts the hardcore bit property~\eqref{eq:adaptive-hardcore}.
\end{proof}

\subsubsection{$Z$-measurement test}

\begin{lemma}\label{lem:zpi}
Let $D$ be an efficient device that succeeds with probability $1$ in the preimage test, and at least $1-\eps$ in the $Z$-measurement test. Then on average over $y\in \mY$,
\begin{equation}\label{eq:zpi-0}
 \sum_{d,b} \Tr \big( (M^d\Pi^b - Z^bM^d)^\dagger (M^d\Pi^b - Z^bM^d)\phi\big) \,\leq\, 2\,\eps+\negl(\lambda)\;.
\end{equation}
\end{lemma}

\begin{proof}
The assumption of success $1-\eps$ in the $Z$-measurement test implies that, on average over $k\in \mK_\mG$ and $y\in \mY$, 
\begin{equation}\label{eq:zpi-1}
\sum_{b,d} \Tr\big(Z^b M^d \Pi^b \phi \Pi^b M^d\big) \,\geq\, 1-\,\eps-\negl(\lambda)\;,
\end{equation}
where we used that for $k\in \mK_\mG$ by Lemma~\ref{lem:preimage} it holds that $\phi = \sum_b\Pi^b \phi \Pi^b$. Since $\Pi$, $M$ and $Z$ can all be efficiently implemented, using the collapsing property (Lemma~\ref{lem:collapse})~\eqref{eq:zpi-1} holds on average over $k\in\mK_\mF$ as well. 

Let $\Pi = \Pi^0-\Pi^1$ act on the first qubit of $\phi$ (written as in~\eqref{eq:preimage-phi}). Again using the collapsing property, $\Pi \phi \Pi $ and $\phi$ are computationally indistinguishable, so
\begin{equation}\label{eq:zpi-2}
\sum_{b,d} \big|\Tr\big(Z^b M^d ( \phi-\Pi \phi \Pi) M^d\big)\big| \,=\, \negl(\lambda)\;.
\end{equation}
Using that $ \phi-\Pi \phi \Pi  = 2(\Pi^0 \phi \Pi^1 + \Pi^1\phi\Pi^0)$, combining~\eqref{eq:zpi-1} and~\eqref{eq:zpi-2} gives 
\[\sum_{b,d} \big(\Tr\big(Z^b M^d \Pi^b \phi M^d\big) + \Tr\big(Z^b M^d \phi \Pi^b M^d\big)\big) \,\geq\, 2\big(1- O(\eps)\big) - \negl(\lambda)\;.\]
Expanding the square in~\eqref{eq:zpi-0}, this proves the lemma. 
\end{proof}

\subsubsection{$X_\theta$-measurement test, part A}

\begin{lemma} \label{lem:break}
Let $D = (\phi,\Pi,M,Z,\{X_\theta\})$ be an efficient device. Define a sub-normalized density 
\begin{align}
\tilde{\phi}_{\reg{YBXD}} &= \sum_{y\in\mY} \proj{y}_\reg{Y}\otimes \sum_{b\in\{0,1\}} \proj{b,x_b}_\reg{BX} \otimes \Pi_y^{(b,x_b)} \,\phi_y \,\Pi_y^{(b,x_b)}\notag\\
&=  \sum_{b\in\{0,1\}} \proj{b,x_b}_\reg{BX} \otimes \tilde{\phi}^{(b)}_\reg{YD}\;.\label{eq:def-sigmay}
\end{align}
Then $\tilde{\phi}_{\reg{YBXD}}$ is the post-measurement state of the device at the end of the preimage test. 
For $v\in\{0,1\}$ and $\theta\in\{0,1,2,3\}$ let  
\begin{align}
\sigma_{\theta,v} &=\sum_{y\in\mY}\sum_{b\in\{0,1\}} \proj{b,x_b}_{\reg{BX}} \otimes  \sum_{d: \hat{\theta}(d)=\theta } \proj{d}\otimes (\Id_\reg{Y}\otimes X_\theta^v M_y^{d}) \tilde{\phi}^{(b)}_{\reg{YD}} (\Id_\reg{Y}\otimes M_y^{d}X_\theta^v)\;.\label{eq:def-rho}
\end{align}
Then for any $\theta\in\{0,1,2,3\}$, $\sigma_{\theta,0}$ and $\sigma_{\theta,1}$ are computationally indistinguishable. 
\end{lemma}

\begin{proof}
The proof is almost identical to the proof of~\cite[Lemma 7.1]{brakerski2018certifiable}. 
Suppose for contradiction that there exists a $\theta\in\{0,1,2,3\}$ and an efficient observable $O$ such that 
\begin{equation}\label{eq:bias-o}
\Tr(O(\sigma_{\theta,0}-\sigma_{\theta,1})) \,\geq\, \mu\;,
\end{equation}
 for some non-negligible function $\mu(\lambda)$. We derive a contradiction with the hardcore bit property~\eqref{eq:adaptive-hardcore}.

Consider the following efficient procedure $\mathcal{A}$. $\mathcal{A}$ first prepares the state $\tilde{\phi}_{\reg{YBXD}}$ in~\eqref{eq:def-sigmay}. This can be done efficiently by first preparing $\phi_{\reg{YD}}$, then measuring a $y\in \mY$, then applying the measurement $\{\Pi_y^{(b,x)}\}$ to $\phi_y$, and returning a special abort symbol if the outcome is invalid, i.e. CHK$_{\mathcal{F}}(k,b,x,y)=0$. 

$\mathcal{A}$ then applies the measurement $\{M_y^{d}\}$ to $\tilde{\phi}_{\reg{YBXD}}$, obtaining an outcome $d\in\Z_8^w$. Next, it measures using $\{X_\theta^v\}$ to obtain $v\in\{0,1\}$. At this point, the procedure has prepared either $\sigma_{\theta,0}$ or $\sigma_{\theta,1}$. Finally, the procedure measures $O$ to obtain a bit $u$, and returns $(b,x,d,\theta,u\oplus v)$. 

This defines an efficient procedure. Using~\eqref{eq:bias-o} it follows  that the procedure violates the hardcore bit property~\eqref{eq:adaptive-hardcore}. To see why, note that the guarantee~\eqref{eq:bias-o} only holds when 
 $\theta = \hat{\theta}(d)$, but this is precisely when~\eqref{eq:adaptive-hardcore} requires that there should be no distinguishing advantage. 
\end{proof}

\begin{corollary}\label{cor:break}
Let $D = (\phi,\Pi,M,Z,\{X_\theta\})$ be an efficient device that succeeds in the preimage test with probability $1$, and in the $Z$-measurement test with probability at least $1-\eps$. Then on average over $y$,
\[\sum_{\theta\in\{0,1,2,3\}}\sum_{b\in\{0,1\}} \Big| \Tr\big( X_{\theta}^0 Z^b \phi_{\theta} Z^b \big) -  \Tr\big( X_{\theta}^1 Z^b \phi_{\theta} Z^b \big) \Big| \,=\,O\big(\sqrt{\eps}\big)\;,\]
where $\phi_\theta$ is defined in Definition~\ref{def:phitheta}.
\end{corollary}

\begin{proof}
Lemma~\ref{lem:break} implies that $\sigma_{\theta,0}$ and $\sigma_{\theta,1}$ must have traces that are negligibly far from each other, i.e.\ for every $\theta$ and on average over $y\in\mY$,
\[\sum_{d:\,\hat{\theta}(d)=\theta}\sum_b \Big|  \Tr\big( X_{\theta}^0 M^d \Pi^b \phi \Pi^b M^d \big) -  \Tr\big( X_{\theta}^1 M^d \Pi^b \phi \Pi^b M^d \big)\Big) \Big| \,=\,\negl(\lambda)\;.\]
Using Lemma~\ref{lem:zpi} and the Cauchy-Schwarz inequality, this expression is within $O(\sqrt{\eps})$ of
\[\sum_{d:\,\hat{\theta}(d)=\theta}\sum_b \Big|  \Tr\big( X_{\theta}^0  Z^b M^d \phi M^d Z^b \big) -  \Tr\big( X_{\theta}^1 Z^b M^d  \phi M^d Z^d\big)\Big) \Big| \,=\,\negl(\lambda)\;,\]
as desired.
\end{proof}

The following lemma shows a strong form of incompatibility between the measurements $Z$ and $X_\theta$, for any efficient device. 

\begin{lemma}\label{lem:ac}
Let $D= (\phi,\Pi,M,Z,\{X_\theta\})$ be an efficient device such that $D$ succeeds with probability $1$ in the preimage test, and with probability at least $1-\eps$ in both the $Z$-measurement test and part A. of the $X_\theta$-measurement test.
 Then there exists $\eps_\ac = O(\eps^{1/4})$ such that on average over $y\in\mY$, 
\[ \sum_{\theta\in\{0,1,2,3\}}\Tr\big(\{Z,X_{\theta}\}^2 \phi_\theta\big) \,\leq\, \eps_\ac\;.\]
\end{lemma}

\begin{proof}
The assumption that $D$ succeeds with probability $1-\eps$ in part A. of the $X_\theta$-measurement test implies that on average  over $y\in\mY$ and $\theta$ distributed according to $\Tr(\phi_\theta)$, 
\begin{equation}\label{eq:ac-1}
 \sum_{v\in\{0,1\}}\,\Tr(X_\theta^v \phi_{\theta,v}) \geq 1-\eps\;.
\end{equation}
Let $\tilde{\phi}_{\theta}$ be the normalized state $\phi_{\theta}/\Tr(\phi_{\theta})$. 
Using Lemma~\ref{lem:d-uniform} to argue that the renormalization is roughly uniform for all but a negligible fraction of all $y$, Corollary~\ref{cor:break} implies that on average over $y$, 
\[\big|\sum_b\Tr( X_\theta^0 Z^b \tilde{\phi}_{\theta} Z^b) - \sum_b\Tr( X_\theta^1 Z^b \tilde{\phi}_{\theta} Z^b)\big|\,=\, O\big(\sqrt{\eps}\big)+\negl(\lambda)\;.\]
Since $\sum_{b,v}\Tr( X_\theta^v Z^b \tilde{\phi}_{\theta} Z^b)=1$, it follows that for any $\theta\in\{0,1,2,3\}$ and $v\in\{0,1\}$,
\begin{equation}\label{eq:ac-2} 
\mu_{\theta,v} = \Big| \frac{1}{2} - \sum_b\Tr( X_\theta^v Z^b \tilde{\phi}_{\theta} Z^b) \Big| \,=\,O\big(\sqrt{\eps}\big)\;.
\end{equation}
Conditions~\eqref{eq:ac-1} and~\eqref{eq:ac-2} place us in a position to apply~\cite[Lemma 7.2]{brakerski2018certifiable}, with $\phi = \phi_{\theta,v}$ (renormalized), $M = X_\theta^v$, and $\Pi = Z^0$. Taking $\omega = \frac{1}{2}+\Omega(\eps^{1/4})$, the lemma implies that the projection $K$ on eigenspaces of the operator
\[\frac{1}{2}\big(ZX_\theta^v Z + X_\theta^v\big) \,=\, Z^0 X_\theta^v Z^0 + Z^1 X_\theta^v Z^1\]
with associated eigenvalue bounded away from $\frac{1}{2}$ by $\Omega(\eps^{1/4})$ satisfies $\Tr((\Id-K)\phi) = O(\sqrt{\eps})$. Thus for $v\in\{0,1\}$, $\Tr(([Z,X_\theta^v]-\frac{1}{2}Z)^2 \phi)=O({\eps}^{1/4})$. The lemma follows.
\end{proof}

Lemma~\ref{lem:ac} specifies that $Z$ and $X_\theta$ are close to anti-commuting on the state $\phi_\theta$. The following lemma uses the collapsing property and the hardcore bit property to argue that anti-commutation extends to any $\phi_u$, for $u\in\{0,1,2,3\}$. 

\begin{lemma}\label{lem:ac-2}
Under the same assumptions as Lemma~\ref{lem:ac}, on average over $y\in\mY$ and for all $\theta\in\{0,1,2,3\}$, 
\[ \sum_{u\in\{0,1,2,3\}}\Tr\big(\{Z,X_{\theta}\}^2 \phi_u\big) \,\leq\, \eps'_{\ac}\;,\]
for some $\eps'_{\ac}=O(\sqrt{\eps_\ac})$.
\end{lemma}

\begin{proof}
First we observe that for a (possibly unknown) $u$ the value $\Tr\big(\{Z,X_{\theta}\}^2 \phi_u\big)$ can be estimated efficiently. This is because for any $\ket{\psi}$, it is possible to implement
\begin{equation}\label{eq:ac-2-1}
 \ket{\psi} \mapsto \frac{1}{\sqrt{2}}\big( \ket{\psi}\ket{0}+\ket{\psi}\ket{1}\big)\mapsto \frac{1}{\sqrt{2}}\big( ZX_\theta \ket{\psi}\ket{0}+X_\theta Z\ket{\psi}\ket{1}\big)\;,
\end{equation}
at which point a measurment of the last qubit in the Hadamard basis returns $\ket{+}$ with probability $\frac{1}{2}\bra{\psi}\{Z,X_{\theta}\}^2 \ket{\psi}$.
By the collapsing property, $\sum_u \Tr\big(\{Z,X_{\theta}\}^2 \phi_u\big)$ is within negligible distance of $\sum_u \Tr\big(\{Z,X_{\theta}\}^2 \tilde{\phi}_u\big)$, where $\tilde{\phi}_u$ is the result of first measuring $\{\Pi^{(b,x_b)}\}$ on $\phi$ and then measuring $\{M^d\}$. 

Now suppose for contradiction that there exists an $u'$ such that $\Tr\big(\{Z,X_{\theta}\}^2 \phi_{u'}\big)$ is noticeably larger than $\Tr\big(\{Z,X_{\theta}\}^2 \phi_u\big)$, for all $u\neq u'$. As argued above, by the collapsing property the same holds with respect to the states $\tilde{\phi}_{u'}$ and $\tilde{\phi}_u$. 

Consider the following efficient procedure $\mA$. Starting from $\phi$, measure $\{\Pi^{(b,x_b)}\}$ to obtain $(b,x_b)$. Then measure $\{M^d\}$ to obtain $d$. Finally, implement the test described in~\eqref{eq:ac-2-1}. If the outcome is $\ket{+}$, return $u'$. If the outcome is $\ket{-}$, repeat a uniformly random $u\in\{0,1,2,3\}$. 

Then $\mA$ returns $(b,x_b,d,u')$ such that $u'=\hat{\theta}(d)$ with probability noticeably larger than $1/4$, violating the adaptive hardcore bit property~\eqref{eq:adaptive-hardcore-2}.
\end{proof}

\subsubsection{$X_\theta$-measurement test, part B}
\label{sec:xpartb}

\begin{lemma}\label{lem:xtheta-rigid}
Let $D= (\phi,\Pi,M,Z,\{X_\theta\})$ be an efficient device, such that $D$ succeeds with probability $1$ in the preimage test, with probability at least $1-\eps$ in both the $Z$-measurement test and part A. of the $X_\theta$-measurement test, and with probability at least $(1-\eps)\opt_{B}$ in part B. of the $X_\theta$-measurement test. 
 Then on average over $y\in\mY$, 
\[ \sum_{u\in\{0,1,2,3\}}\frac{1}{2}\sum_{\theta\in\{0,1\}}\Tr\big(\{X_\theta,X_{\theta+2}\}^2 \phi_u\big) \,\leq\, \eps''_\ac\;,\]
for some $\eps''_\ac = O((\eps'_\ac)^{1/4})$.
\end{lemma}

\begin{proof}
We perform a reduction to Lemma~\ref{lem:qrac}. 
The main work we need to do is argue that the observables $X_0$ and $X_{2}$ can be represented as observables acting on the same qubit. (The case of $X_{1}$ and $X_{3}$ is similar.)

Applying Lemma~\ref{lem:ac-2} for $\theta=0$ and $\theta=2$ followed by Lemma~\ref{lem:pauli-c} we deduce that there is an isometry $V:\mH_\reg{D}\to \C^2 \otimes \mH_{{\reg{D}'}}$ and $\delta' = O(\sqrt{\eps_\ac})$ under which which $Z\simeq_{\delta'} \sigma_Z\otimes \Id$, $X_0 \simeq_{\delta'} \sigma_X \otimes \Id$, and $X_{2}\simeq_{\delta'} \sigma_X \otimes A_X + \sigma_Y \otimes A_Y$, where $A_X$, $A_Y$ are Hermitian commuting such that $A_X^2 + A_Y^2 = \Id$. 

Let $\{\ket{v_j}\}$ be a joint diagonalization basis of $A_X$ and $A_Y$. Let $\rho = \sum_{u\in\{0,1,2,3\}} V\phi_u V^\dagger$, as a density matrix on $\C^2 \otimes \mH_{{\reg{D}'}}$ ($\rho$ implicitly depends on $y$, so it is not normalized). Define a distribution $p_{j,y} = \Tr(\rho^{(j)})$, with $\rho^{(j)}$ the single-qubit density matrix  
\begin{equation}\label{eq:xtheta-1}
\rho^{(j)} = (\Id \otimes \bra{v_j}) \rho (\Id \otimes \ket{v_j})\;.
\end{equation} 
For any $j,y$ define a $2\mapsto 1$ QRAC as follows. The encoding of $u+4v \in\{1,3,5,7\}$, with $u\in\{1,3\}$ and $v\in\{0,1\}$, is the renormalized density matrix $\rho^{(j)}_{u,v}$, defined as $\rho^{(j)}$ in~\eqref{eq:xtheta-1} with $\rho_{u,v}$ instead of $\rho$. The observables are $\sigma_X$ and $\bra{v_j}A_X\ket{v_J} \sigma_X + \bra{v_j}A_Y\ket{v_j}\sigma_Y$. 

Using Lemma~\ref{lem:d-uniform} and the assumption on the device's success probability in part B. of the $X_\theta$-measurement test it follows that this QRAC, on average over $(y,j)$, satisfies the assumption of Lemma~\ref{lem:qrac}, for some $\delta = O(\sqrt{\delta'})$. The conclusion follows. 
\end{proof}

We end with the proof of the main lemma of the section, Lemma~\ref{lem:rigidity}. 

\begin{proof}[Proof of Lemma~\ref{lem:rigidity}]
Fix a strategy for the prover that is accepted with probability at least $1-\eps$ in the qubit preparation test. Then the strategy is accepted with probability at least $1-2\eps$ in the preimage test, and at least $1-4\eps$ in each of the $Z$-measurement test and the $X_\theta$-measurement test. 

Applying Lemma~\ref{lem:ac-2} and Lemma~\ref{lem:xtheta-rigid} followed by Lemma~\ref{lem:pauli-c} to $Z$, $X_0$ and $X_2$  if follows that there exists an efficient isometry $V$ from $\mH_B$ to $\C^2 \otimes \mH_{B'}$ under which $Z \simeq_{\delta_1} \sigma_Z \otimes \Id$, $X_0 \simeq_{\delta_1} \sigma_{X}$ and $X_2 \simeq_{\delta_1} \sigma_Y \otimes \Id$, for some $\delta_1 = O(\sqrt{\eps''_{ac}})$. It follows from success in part (b)(ii)A. that, under the isometry, for $\theta\in\{0,2\}$ and $v\in\{0,1\}$ the state $\phi_{\theta,v}$ is within $O(\delta_1)$ of $\proj{+_{\theta\frac{\pi}{4}+v\pi}} \otimes \proj{\aux_{\theta,v}}$, for some states $\ket{\aux_{\theta,v}}$. Using success in part (b)(ii)B. (the QRAC test) and applying Lemma~\ref{lem:qrac} it similarly follows that for $\theta \in \{1,3\}$ and $v\in\{0,1\}$ the state $\phi_{\theta,v}$ is, under the same isometry, within $O(\delta_1)$ of a state of the form $\proj{+_{\theta\frac{\pi}{4}+v\pi}} \otimes \proj{\aux_{\theta,v}}$. 

Recall that by Lemma~\ref{lem:d-uniform} the states $\phi_{\theta}$ are computationally indistinguishable. According to the previous paragraph, 
\begin{align*}
 \phi_\theta &\simeq_{\delta_1} \frac{1}{2}(\proj{0}  + \proj{1} \big)\otimes \big(\proj{\aux_{\theta,0}} + \proj{\aux_{\theta,1}} \big) \\
&\qquad +  \big(e^{i\theta\frac{\pi}{4}} \ket{0}\bra{1} + e^{-i\theta\frac{\pi}{4}} \ket{1}\bra{0}\big) \otimes \big(\proj{\aux_{\theta,0}} - \proj{\aux_{\theta,1}} \big) \;.
\end{align*}
Since the operators $e^{i\theta\frac{\pi}{4}} \ket{0}\bra{1} + e^{-i\theta\frac{\pi}{4}} \ket{1}\bra{0}$ have constant trace distance for distinct values of $\theta$, it follows that $\proj{\aux_{\theta,0}} \approx_c \proj{\aux_{\theta,1}}$ for all $\theta$, and that they are computationally indistinguishable for different values of $\theta$.

This gives the second condition in the lemma. To obtain the first, recall that under $V$ it holds that $Z \simeq_{\delta_1} \sigma_Z \otimes \Id$. Using success in part (b)(i) of the protocol the state at the beginning of step (b)(i) is of the form $\sum_b \proj{b} \otimes \proj{b} \otimes \proj{\aux_b}$, where the first $b$ is held by the verifier and $\ket{\aux_b}\in\mH_{B'}$ are abitrary (not necessarily normalized). Using the collapsing property (Lemma~\ref{lem:collapse}), for $b\in\{0,1\}$, $\proj{\aux_b}$ is computationally indistinguishable from any of the $\proj{\aux_{\theta,0}}+\proj{\aux_{\theta,1}}$.
\end{proof}

\subsection{Real protocol for remote state preparation}
\label{sec:protocol}

In this section we introduce a many-round protocol that repeatedly calls the qubit preparation test. Eventually, the protocol returns either ``abort'', or an angle $\theta \in \Theta$. The protocol is described in Figure~\ref{fig:protocol}.

\begin{figure}[htbp]
\rule[1ex]{16.5cm}{0.5pt}\\
Let $\lambda$ be a security parameter, $N\geq 1$ a maximum number of rounds, $\delta$ an error tolerance parameter, and $W\in\{X,Z\}$ a basis choice.\\
At the start of the protocol, the verifier communicates $N$ to the prover. The verifier privately samples a number of rounds $R \leftarrow_U \{1,\ldots,N\}$.
\begin{enumerate}
\item For $i=1,\ldots,R$, the verifier executes the qubit preparation test, Figure~\ref{fig:qubit}. They record the outcome of the verifier in the test: either $pass$, or $fail_p$, $fail_Z$, $fail_{X}$ or $fail_Q$. 
\item The verifier sets $flag\leftarrow abort$ if any of the following conditions is satisfied:
\begin{enumerate}
\item The fraction of preimage tests that returned $flag=fail_p$ is larger than $\delta$;
\item The fraction of $Z$-measurement tests that returned $flag=fail_Z$ is larger than $\delta$;
\item The fraction of $X_\theta$-measurement tests, part A., that returned $flag=fail_X$ is larger than $\delta$;
\item The fraction of $X_\theta$-measurement tests, part B., that returned $flag=fail_Q$ is larger than $(1-\opt_Q)+\delta$.
\end{enumerate}
\item[] If $flag= abort$, the verifier aborts and sends the message $ERR$ to the prover. 
\item The verifier samples a key $(k,t_k)\leftarrow \Gen_\mF(1^\lambda)$ (if $W=Z$) or $(k,t_k)\leftarrow \Gen_\mG(1^\lambda)$ (if $W=X$), sends $k$ to the prover and keeps the trapdoor information $t_k$ private. 
\item The prover returns a  $y \in \mY$ to the verifier. 
\item  The verifier requests an equation $d\in \Z_8^w$. If $W=Z$ the verifier computes $(b,x_b)\leftarrow\Inv_\mG(t_k,y)$ and returns $b$. If $W=X$ the verifier computes $(\hat{\theta},\hat{v}) = (\hat{\theta}(d),\hat{v}(d))$ and returns $\theta =  \hat{\theta}\frac{\pi}{4} + v\pi$. 
\end{enumerate}
\rule[1ex]{16.5cm}{0.5pt}
\caption{The remote state preparation protocol. See Section~\ref{sec:entcf} for notation associated with the extended NTCF family $\mathcal{F}$.}
\label{fig:protocol}
\end{figure}

To show that the protocol can be used to implement the \RSPV\ resource, which will be done in the next section, we prove the following.  


\begin{theorem} \label{thm:rigidity}
Suppose the remote state preparation protocol (Figure~\ref{fig:protocol}) with security parameter $\lambda > 0$, maximum number of rounds $N\in\N$ and error tolerance $\delta\in[0,1]$ is executed with an arbitrary quantum polynomial-time prover. Assume that the protocol succeeds with probability at least $\omega$, for some $\omega>0$ such that $N \geq \delta^{-3}\log(2/\delta\omega)$. Then there exists an efficient isometry $\Phi$ and a state $\ket{\aux}$ such that the joint state of the verifier's input bit $W$, his output angle $b$ ($W=Z$) or $\theta$ ($W=X$), and the prover's final state, conditioned on not aborting in step 2.,  is such that 
\begin{multline} \label{eqn:rigidity}
\rho_{\reg{S}\Theta \reg{B}} \,\simeq_\eps\,  p_Z \proj{Z}_\reg{S} \otimes \frac{1}{2}\sum_{b\in \{0,1\}} \proj{b}_{\Theta} \otimes \Phi(\ket{b} \ket{\aux} \big)_\reg{B}  \\
+ p_X \proj{X}_\reg{S} \otimes \frac{1}{8}\sum_{\theta\in \Theta} \proj{\theta}_{\Theta}\otimes \Phi( \ket{+_{\theta}} \ket{\aux})_\reg{B} \;,
\end{multline} 
where $\varepsilon = O(\delta^c) + negl(\lambda)$, for some constant $c > 0$, and $p_Z$, $p_X$ are the verifier's prior probability of choosing $W=Z$ or $W=X$ respectively. Moreover, the honest strategy introduced in the proof of Lemma~\ref{lem:qubit-completeness} succeeds with probability negligibly close to $1$ in the protocol. At the end of the protocol, the verifier returns a uniformly random $b\in\{0,1\}$ and the prover's state is $\ket{b}$ ($W=Z$) or $\theta\in\Theta$ ($W=X$) and the prover's state is $\ket{+_\theta}$. 
\end{theorem}

\begin{proof}
The completeness property follows from Lemma~\ref{lem:qubit-completeness} and a standard concentration bound. 

To show soundness, fix a strategy for the prover that succeeds with probability at least $\omega$. For $i\in\{1,\ldots, N\}$ let $T_i \in \{0,1\}$ be a random variable that equals $1$ if and only if the prover does not cause the verifier to raise a $fail$ flag in the $i$-th round. By assumption on the prover's success probability it holds that $T=\frac{1}{R}\sum_{i=1}^R T_i \geq (1-\delta)\opt$ with probability at least $\omega$, where $\opt = \frac{3}{4}+\frac{1}{4}\opt_Q$. 
 
Applying Azuma's inequality, the probability that $T$ deviates from its expectation by more than $\delta$ is at most $2e^{-\delta^2 R/2}$. It follows that as long as 
\begin{equation}\label{eq:cond-r}
2e^{-\delta^2 R/2}\omega\leq \delta\;,
\end{equation}
the expectation of $T$ conditioned on success satisfies $\Es{}[T|\text{not }abort] \geq (1-\delta)\opt$.  Under the assumption on $N$, $\delta$ and $\omega$ made in the theorem, condition~\eqref{eq:cond-r} is satisfied with probability at least $1-\delta$ over the choice of $R$. Applying Markov's inequality, a randomly chosen round satisfies $\Es{}[T_i|\text{not }abort] \geq 1-O(\sqrt{\delta})$ with probability at least $1-O(\sqrt{\delta})$. Provided such a round is chosen as the $R$-th round, we can apply Lemma~\ref{lem:rigidity} to conclude. 
\end{proof}

%% file: composable.tex
In this section we show that the remote state preparation protocol introduced in Section~\ref{sec:real-protocol} constructs the ideal \RSPV\ resource described in Section~\ref{sec:intro}.
The definition of \RSPV\ (illustrated in Figure~\ref{fig:rspv}) is as follows:
\begin{definition}[Random Remote State Preparation with Verification] \label{def:rspv.ideal}
The resource receives $W\in\{X,Z\}$ from Alice's interface and the bit $c\in \{0, 1\}$ from Bob's interface. If $c=0$ and $W=Z$, Alice receives a uniformly random bit $b\in\{0,1\}$ and Bob receives the state $\ket{b}$. If $c=0$ and $W=X$ Alice receives a uniformly random value $\theta \in \Theta = \{0,\frac{\pi}{4},\ldots,\frac{7\pi}{4}\}$ and Bob receives the state $\ket{+_\theta}$. If $c=1$ both Alice and Bob receive an $ERR$ message. 
\end{definition}

Recall that our goal is to show that an implementation of the RSP protocol described in the previous section, based on a classical channel and a measurement buffer as communication resources,\footnote{The measurement buffer is used each time the measurement test, step 3(b) of the qubit preparation test in Figure~\ref{fig:qubit}, is executed. The classical channel is used for all other steps.} securely implements \RSPV. In the AC language, the implementation of RSP using the communication resources is known as the real protocol, and we denote it BRSP, for buffered remote state preparation. As an abstract functionality, we illustrate it in Figure~\ref{fig:prot.rsp}. Alice has an input $W$,\footnote{Alternatively, we say that Alice and Bob receive their inputs from (and return their outputs to) an environment.} specifying the basis for state preparation, and produces as output either a random $b$ if $W=Z$, a random $\theta$ if $W=X$, or $ERR$ if Bob behaved maliciously. Bob takes as input the bit $c$, specifying whether he should behave honestly or maliciously. Mirroring Alice, his output is either a state $\ket{b}$, a state $\ket{+_{\theta}}$, or $ERR$. The two interact via a classical channel and a measurement buffer according to the specification given in Figure~\ref{fig:prot.rsp}. 

\begin{figure}[htb]
\begin{centering}

\begin{tikzpicture}[opnode/.style={minimum width=1.05cm,minimum height=1.0cm}]
\small

\def\t{5}
\def\u{2.2}
\def\v{1.2}

\node[opnode] (a1) at (-\u,0) {};
\node[opnode] (a2) at (-\u,-\v) {};
\node[opnode] (a3) at (-\u,-2*\v) {$\vdots$};
\node[opnode] (a35) at (-\u,-2.12*\v) {};
\node[opnode] (a4) at (-\u,-3*\v) {};
\node[opnode] (a5) at (-\u,-3.5*\v) {};

\node[opnode] (aa1) at (-\u,0.25) {};
\node[opnode] (aa11) at (-\u,-0.25) {};
\node[opnode] (aa2) at (-\u,-\v+0.25) {};
\node[opnode] (aa22) at (-\u,-\v-0.25) {};
\node[opnode] (aa5) at (-\u,-3.5*\v+0.25) {};
\node[opnode] (aa55) at (-\u,-3.5*\v-0.25) {};

\node[draw,inner sep=.25cm,fit=(a1)(a5)] (a) {};
\node[yshift=-2,above right] at (a.north west) {$\pi_A$};
\node (alice1) at (-\t + 1,0) {};
\node (alice11) at (-\t,-3*\v) {};
\node (alice2) at (-\t,-3.5*\v) {};
\node (alice) at (-\t-.4,-1.75*\v) {\large \textbf{Alice}};

\node[opnode] (b1) at (\u,0) {};
\node[opnode] (b2) at (\u,-\v) {};
\node[opnode] (b3) at (\u,-2*\v) {$\vdots$};
\node[opnode] (b35) at (\u,-2.12*\v) {};
\node[opnode] (b4) at (\u,-3*\v) {};
\node[opnode] (b5) at (\u,-3.5*\v) {};

\node[opnode] (bb1) at (\u,0.25) {};
\node[opnode] (bb11) at (\u,-0.25) {};
\node[opnode] (bb2) at (\u,-\v+0.25) {};
\node[opnode] (bb22) at (\u,-\v-0.25) {};
\node[opnode] (bb5) at (\u,-3.5*\v+0.25) {};
\node[opnode] (bb55) at (\u,-3.5*\v-0.25) {};

\node[draw,inner sep=.25cm,fit=(b1)(b5)] (b) {};
\node[yshift=-2,above right] at (b.north west) {$\pi_B$};
\node (bob1) at (\t - 1,0) {};
\node (bob11) at (\t,-3*\v) {};
\node (bob2) at (\t,-3.5*\v) {};
\node (bob) at (\t+.4,-1.75*\v) {\large \textbf{Bob}};

\node[opnode] (c1) at (0,0) {$\mathcal{R}$};
\node[opnode] (c2) at (0,-\v) {$\pi_{Buff}$};
\node[opnode] (c35) at (0,-2*\v) {$\vdots$};
\node[opnode] (c4) at (0,-3*\v) {};
\node[opnode] (c5) at (0,-3.5*\v) {$\pi_{Buff}$};

\node[opnode] (cc1) at (0,0.25) {};
\node[opnode] (cc11) at (0,-0.25) {};
\node[opnode] (cc2) at (0,-\v+0.25) {};
\node[opnode] (cc22) at (0,-\v-0.25) {};
\node[opnode] (cc5) at (0,-3.5*\v+0.25) {};
\node[opnode] (cc55) at (0,-3.5*\v-0.25) {};

\node[draw,minimum width=1.6cm,inner sep=.03cm,fit=(c1)(c1)] (c) {};
\node[draw,minimum width=1.6cm,inner sep=.03cm,fit=(c2)(c2)] (c) {};
\node[draw,minimum width=1.6cm,inner sep=.03cm,fit=(c5)(c5)] (c) {};

\draw[sArrow] (aa1) to (cc1);
\draw[sArrow] (cc11) to (aa11);
\draw[sArrow] (cc1) to (bb1);
\draw[sArrow] (bb11) to (cc11);

\draw[sArrow] (aa2) to (cc2);
\draw[sArrow] (cc22) to (aa22);
\draw[sArrow] (bb2) to (cc2);
\draw[sArrow] (cc22) to (bb22);
\draw[sArrow] (b2) to (c2);

\draw[sArrow] (aa5) to (cc5);
\draw[sArrow] (cc55) to (aa55);
\draw[sArrow] (bb5) to (cc5);
\draw[sArrow] (cc55) to (bb55);
\draw[sArrow] (b5) to (c5);


\draw[sArrow] (alice1.center) to node[auto,pos=.4] {$W$} (a1);
\draw[sArrow] (bob1.center) to node[auto,swap,pos=.4] {$c$} (b1);
\draw[sArrow] (a5) to node[auto,pos=.6] {$b/\theta/ERR$} (alice2.center);
\draw[sArrow] (b5) to node[auto,swap,pos=.6] {$\;\; \ket{b/+_{\theta}/ERR}$} (bob2.center);
\end{tikzpicture}

\end{centering}
\caption{\label{fig:prot.rsp} The remote state preparation protocol, illustrated here schematically as an AC functionality. Alice and Bob interact with the measurement buffer (which behaves as described in Figure~\ref{fig:meas.buffer}), $\pi_{Buff}$, and the classical channel, $\mathcal{R}$, for a number of rounds. At the end of the interaction, upon success Alice obtains a bit $b$ or an angle $\theta$ and Bob obtains the state $\ket{b}$ or the state $\ket{+_{\theta}}$. Otherwise, both parties obtain $ERR$. While not explicitly shown in the figure, Bob exchanges both classical and quantum messages with the buffer, whereas Alice interacts only classically with the buffer and $\mathcal{R}$.}
\end{figure}
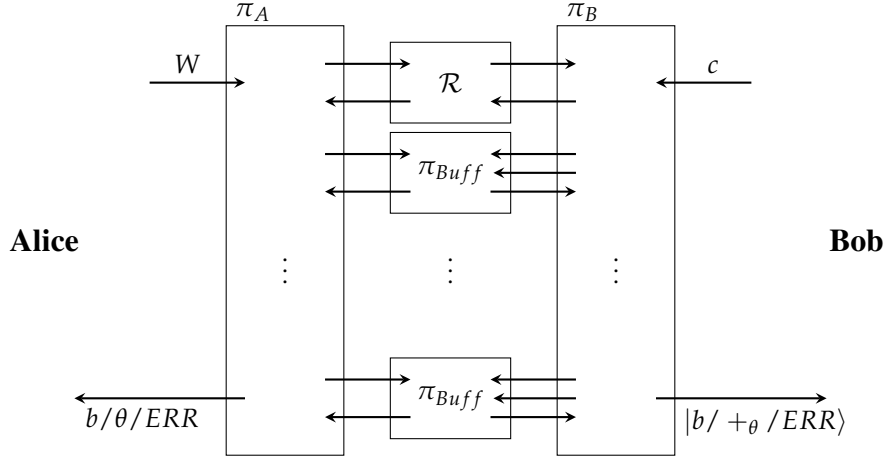

We note that in BRSP the buffer is only needed for step 3(b) of the protocol --- from Figure~\ref{fig:protocol} --- in which there is a constant number of challenges, so that it remains efficient for Bob to forward a specification of each of its measurements. The fact that we build the real protocol from a measurement buffer is necessary for the security proof to go through. Informally, and outside of the AC framework, the measurement buffer is ``without loss of generality'': in any execution of the protocol, Bob's answer to a challenge from Alice is obtained by making a measurement on a quantum state; since we assume that Bob is computationally efficient an explicit description of the measurement exists, and this is all that is needed for the stand-alone security proof.

Obtaining composable security is more subtle, and this is why the buffer is needed. 
Note that its use does \emph{not} preclude Bob from sharing a prior entangled state with the environment, nor from exchanging quantum messages with the environment in-between any two uses of the measurement buffer. In this sense the use of the buffer is comparable yet much less restrictive to the way a ``device'' is defined to obtain composable security of device-independent protocols for e.g. randomness expansion~\cite{portmann_talk}. In that context, Alice (the verifier) interacts with a device that is prepared by Bob (the eavesdropper). Bob is allowed to provide the initial state of the device, but he does not interact with it at any later stage of the protocol, and in particular is not allowed to receive the contents of the internal memory of the device at the end of the protocol. In our setting, this latter point is allowed.

Summarizing, we view the protocol as consisting of three parties $\pi = (\pi_A, \pi_{Buff}, \pi_B)$, where $\pi_A$ denotes Alice's actions, $\pi_B$ denotes Bob's actions and $\pi_{Buff}$ the actions of  the measurement buffer. Whenever step 3 of the protocol is executed Bob sends his state and measurements to the buffer, where it is measured according to the measurement specified by Alice's challenge. The measurement result is sent to both Alice and Bob. Bob, in addition, receives Alice's challenge and the post-measurement state. We now show that BRSP constructs the ideal \RSPV\ resource from classical channels.

\begin{theorem} \label{thm:brsp.rspv}
The buffered remote state preparation protocol implements the ideal \RSPV\ functionality. In other words, let $\mathcal{R}$ denote a classical channel and $\pi = (\pi_A, \pi_{Buff}, \pi_B)$ denote the BRSP protocol:
\begin{itemize}
\item $\pi_A$ takes as input a bit $W$, specifying either the $Z$ basis or the $X$ basis and produces as output either a bit $b$, an angle $\theta$ or the $ERR$ flag. The actions that Alice performs in $\pi_A$ are exactly the same as in the RSP protocol, as described in Figure~\ref{fig:prot.rsp}. $\pi_A$ interacts with the classical channel $\mathcal{R}$.
\item $\pi_B$ takes as input a bit $c$, specifying either honest (when $c=0$) or malicious (when $c=1$) behavior and outputs either $\ket{b}$, $\ket{+_{\theta}}$ or the $ERR$ flag. It ignores the bit $c$ and always behaves honestly, i.e. perform the actions instructed by Alice in RSP. For each measurement test performed, $\pi_B$ sends the state to be measured to $\pi_{Buff}$ and expects to receive the measurement outcome and the post-measurement state.
\item $\pi_{Buff}$ receives a message from Alice. It receives from Bob a specification $\mathcal{F}$ of a measurement to perform for each of Alice's messages, as well as a state $\rho$ to be measured. It measures the state according to $\mathcal{F}(M)$. The measurement outcome is returned to both Alice and Bob. The post-measurement state is returned to Bob.
\end{itemize}
We denote by $\perp_B$ a filtered functionality for Bob that has him set $c = 0$ in \RSPV. Additionally, let $\lambda > 0$ denote the security parameter used in BRSP, and $\delta > 0$ denote the error tolerance parameter of BRSP. Then it holds that
\begin{equation} \label{eqn:brsp.correctness}
\pi_A \mathcal{R} \pi_{Buff} \pi_B \approx_{c, \varepsilon_1} \RSPV \perp_B\;,
\end{equation}
\noindent and there exists a polynomial-time quantum simulator $\sigma_B$ such that
\begin{equation} \label{eqn:brsp.security}
\pi_A \mathcal{R} \pi_{Buff} \approx_{c, \varepsilon_2} \RSPV \,\sigma_B\;,
\end{equation}
where $\varepsilon_1 = \negl(\lambda)$ and $\varepsilon_2 = O(\delta^c) + \negl(\lambda)$, for some constant $c > 0$.
\end{theorem}

\begin{proof}
Eq.~\eqref{eqn:brsp.correctness} follows immediately from the completeness of BRSP, that is inherited from the completeness of RSP (see Theorem~\ref{thm:rigidity}). Indeed, if Alice, Bob and the buffer follow the protocol, their results are exactly those obtained in the ideal functionality. 

For Eq.~\eqref{eqn:brsp.security}, note that we are assuming that Bob is the only malicious party, whereas Alice and the buffer still follow their honest actions in the protocol. 
Let us consider a simulator $\sigma_B$ that executes the BRSP protocol with Bob. More specifically, the simulator executes the ``buffered'' analogue of the  protocol from Figure~\ref{fig:protocol} with Bob. 

If the protocol aborts before step $3$, the simulator sets $c=1$ in the ideal \RSPV\ functionality, indicating $ERR$. If the protocol does not abort, the simulator chooses uniformly at random whether to perform the $W_S=Z$ or the $W_S=X$ run of the protocol. (We denote by $W_S$ the simulator's choice of basis to avoid confusion with $W$, which denotes Alice's choice.)
In step $4$, the simulator takes Bob's state and performs the measurement that returns the string $y$. 
Note that at step 5, from the point of view of Bob, the situation is indistinguishable from step 3(b) in the qubit preparation test (Figure~\ref{fig:qubit}). The simulator requests measurements for Bob associated with that step, that is, a ``preimage measurement'' as well as measurements $Z$ and $X_\theta$ associated with part (b)(ii). (To ensure that Bob does not detect this modification, we may assume that in the buffered protocol execution Bob's measurements for step 3 are always requested at the end of step 2, irrespective of whether step 3 is performed or not, i.e.\ prior to round $R$ or not.)

If the simulator does not receive a state and measurements of matching dimension it sets $c=1$ and causes the ideal functionality to abort.
Assuming the simulator has not aborted after receiving Bob's final state, 
it sets $c=0$ in the \RSPV\  functionality, and takes the resulting state. Using the specification of Bob's measurement operators $Z$ and $X_\theta$ it computes the isometry $\Phi$ whose existence is guaranteed by Theorem~\ref{thm:rigidity},\footnote{As long as the measurements are observables of the same dimension the isometry is always well-defined, whether the assumptions of the theorem are satisfied or not.} ``undoes'' the isometry by applying its inverse, replaces the first qubit of the $\reg{B}$ register in~\eqref{eqn:rigidity} by the qubit obtained from \RSPV, and re-applies the isometry. It returns the resulting state to Bob. 
Let us now consider a distinguisher that interacts with either $\pi_A \mathcal{R} \pi_{Buff}$ or $\RSPV \sigma_B$ and which has the initial state $\psi$.
In the first case we  denote the distinguisher's final state as $\tau^{\psi}_{AB}$. Assuming that in the protocol Alice accepts with probability $1 - p^{\psi}$, for some $p^{\psi} > 0$, we have that: 
\begin{equation} \label{eqn:brsp.tau}
\tau^{\psi}_{AB} = (1 - p^{\psi}) \rho^{\psi}_{AB} + p^{\psi} \ket{ERR}\bra{ERR}_A \otimes \gamma^{\psi}_B
\end{equation}
where
\begin{multline} \label{eqn:brsp.taugood}
\rho^{\psi}_{AB} \,\simeq_{\eps_2}\,  p^{\psi}_Z \proj{Z}_A \otimes \frac{1}{2}\sum_{b\in \{0,1\}} \proj{b}_A \otimes \Phi^{\psi}(\ket{b} \ket{\aux^{\psi}} \big)_B + \\
+  p^{\psi}_X \proj{X}_A \otimes \frac{1}{8}\sum_{\theta\in \Theta} \proj{\theta}_A \otimes \Phi^{\psi}( \ket{+_{\theta}} \ket{\aux^{\psi}}))_B 
\end{multline}
and $\gamma^{\psi}_B$ is some state on Bob's side that is consistent with the protocol having aborted and the initial state of the system being $\psi$.
The expression from Eq.~\eqref{eqn:brsp.taugood} follows from the rigidity theorem (Theorem~\ref{thm:rigidity}), since, conditioned on success, the reduced state of Alice and Bob takes the form in~\eqref{eqn:brsp.taugood}. Note that the probabilities for the choices $W=Z$ and $W=X$, respectively, as well as the isometry on Bob's system and his auxiliary state, are determined by the initial state $\psi$.

Now let us consider the case when the distinguisher interacts with the ideal functionality and the simulator. We denote the final state, \emph{prior to the simulator performing the qubit swap}, as $\sigma_{ASB}^{\psi}$. Once again, assuming the probability of acceptance for BRSP (as run by the simulator this time) is $1 - p^{\psi}$, with $p^{\psi} > 0$, we have that:
\begin{equation*}
\sigma^{\psi}_{ASB} = (1 - p^{\psi}) \rho^{\psi}_{ASB} + p^{\psi} \ket{ERR}\bra{ERR}_A \otimes \ket{ERR}\bra{ERR}_S \otimes \gamma^{\psi}_B
\end{equation*}
where $\gamma^{\psi}_B$ is the same as in~\eqref{eqn:brsp.tau}. At this point in the protocol, $\rho^{\psi}_{ASB} = \rho^{\psi}_A \otimes \rho^{\psi}_{SB}$, since, conditioned on acceptance in BRSP, the simulator has not yet interacted with the ideal functionality.
Using the rigidity theorem again, we get
\begin{multline*}
\rho^{\psi}_{SB} \,\simeq_{\eps_2}\,  \frac{1}{2} \proj{Z}_S \otimes \frac{1}{2}\sum_{b\in \{0,1\}} \proj{b}_S \otimes \Phi^{\psi}(\ket{b}_S \ket{\aux^{\psi}}_{SB} \big) + \\
+  \frac{1}{2} \proj{X}_S \otimes \frac{1}{8}\sum_{\theta\in \Theta} \proj{\theta}_S \otimes \Phi^{\psi}( \ket{+_{\theta}}_S \ket{\aux^{\psi}}_{SB})) \;.
\end{multline*}
Note that the probabilities for the $Z$ and $X$ tests are equal, since we've established that the simulator chooses which to perform uniformly at random.
Also note that the $\ket{\aux^{\psi}}$ system is shared between the simulator and Bob.
The simulator now sets $c=0$ in the ideal functionality and receives the state it provides. As already described, it  then ``undoes'' the isometry on the state it has from Bob, replaces the first qubit with the one from the ideal functionality, reapplies the isometry and sends the state to Bob.

As in the interaction with the real protocol, supposing that Alice (as controlled by the distinguisher) chooses to perform the two tests with probabilities $p_Z^{\psi}$ and $p_X^{\psi}$, respectively, the output she  receives from the ideal functionality is either $b$ or $\theta$, with the associated   probabilities. The state that the simulator receives and swaps into Bob's system is of course classically correlated with this output, being either $\ket{b}$ or $\ket{+_{\theta}}$.
If we now write the state of the system upon the completion of this last step we have
\begin{equation*}
\tilde{\sigma}^{\psi}_{ASB} = (1 - p^{\psi}) \tilde{\rho}^{\psi}_{ASB} + p^{\psi} \ket{ERR}\bra{ERR}_A \otimes \ket{ERR}\bra{ERR}_S \otimes \gamma^{\psi}_B\;,
\end{equation*}
where
\begin{equation*}
\tilde{\rho}_{ASB}^{\psi} \,\simeq_{\eps_2}\, p_{Z}^{\psi} \proj{Z}_A \otimes \frac{1}{2} \sum_{b \in \{0, 1\}} \proj{b}_A \otimes \zeta^{\psi}_{SB}(b) + p_{X}^{\psi} \proj{X}_A \otimes \frac{1}{8} \sum_{\theta \in \Theta} \proj{\theta}_A \eta^{\psi}_{SB}(\theta)
\end{equation*}
and $\zeta^{\psi}_{SB}(b)$, $\eta^{\psi}_{SB}(\theta)$ are states on the joint system of the simulator and Bob, given by:
\begin{multline*}
\zeta^{\psi}_{SB}(\mathbf{b}) \,\simeq_{\eps_2}\, \frac{1}{2} \proj{Z}_S \otimes \frac{1}{2}\sum_{b' \in \{0,1\}} \proj{b'}_S \otimes \proj{b'}_S \otimes \Phi^{\psi}(\mathbf{\ket{b}} \ket{\aux^{\psi}} \big)_B + \\
\frac{1}{2} \proj{X}_S \otimes \frac{1}{8}\sum_{\theta' \in \Theta} \proj{\theta'}_S \otimes \proj{+_{\theta'}}_S \otimes \Phi^{\psi}( \mathbf{\ket{b}} \ket{\aux^{\psi}}))_B \;,
\end{multline*}
\begin{multline*}
\eta^{\psi}_{SB}(\boldsymbol{\theta}) \,\simeq_{\eps_2}\, \frac{1}{2} \proj{Z}_S \otimes \frac{1}{2}\sum_{b' \in \{0,1\}} \proj{b'}_S \otimes \proj{b'}_S \otimes \Phi^{\psi}(\ket{\pmb{+}_{\pmb{\theta}}} \ket{\aux^{\psi}} \big)_B + \\
\frac{1}{2} \proj{X}_S \otimes \frac{1}{8}\sum_{\theta' \in \Theta} \proj{\theta'}_S \otimes \proj{+_{\theta'}}_S \otimes \Phi^{\psi}( \ket{\pmb{+}_{\pmb{\theta}}} \ket{\aux^{\psi}}))_B \;.
\end{multline*}
The boldface letters highlight the state that was planted in Bob's system.
Now notice that if we trace out the simulator's system from both of these states we get:
\begin{equation*}
Tr_{S}(\zeta^{\psi}_{SB}) = \zeta^{\psi}_{B}(\mathbf{b}) \,\simeq_{\eps_2}\, \Phi^{\psi}( \mathbf{\ket{b}} \ket{\aux^{\psi}}))_B \;,
\end{equation*}
\begin{equation*}
Tr_{S}(\eta^{\psi}_{SB}) = \eta^{\psi}_{B}(\boldsymbol{\theta}) \,\simeq_{\eps_2}\, \Phi^{\psi}( \ket{\pmb{+}_{\pmb{\theta}}} \ket{\aux^{\psi}}))_B \;.
\end{equation*}
Tracing out the simulator from $\tilde{\rho}_{ASB}$ and plugging in the above expressions we obtain
\begin{multline*}
\tilde{\rho}_{AB}^{\psi} \,\simeq_{\eps_2}\, p_{Z}^{\psi} \proj{Z}_A \otimes \frac{1}{2} \sum_{b \in \{0, 1\}} \proj{b}_A \otimes \Phi^{\psi}( \ket{b} \ket{\aux^{\psi}}))_B  + \\ 
+ p_{X}^{\psi} \proj{X}_A \otimes \frac{1}{8} \sum_{\theta \in \Theta} \proj{\theta}_A \Phi^{\psi}( \ket{+_{\theta}} \ket{\aux^{\psi}}))_B \;.
\end{multline*}
Finally, if we trace out the simulator from $\tilde{\sigma}_{ASB}^{\psi}$ and use the above state, we have:
\begin{equation*}
\tilde{\sigma}^{\psi}_{AB} = (1 - p^{\psi}) \tilde{\rho}^{\psi}_{AB} + p^{\psi} \ket{ERR}\bra{ERR}_A \otimes \gamma^{\psi}_B\;.
\end{equation*}
From the triangle inequality $\rho_{AB}^{\psi} \,\simeq_{2\eps_2}\, \tilde{\rho}_{AB}^{\psi}$, and therefore we have that:
\begin{equation*}
\pi_A \mathcal{R} \pi_{Buff} \approx_{c, 2\varepsilon_2} \RSPV \sigma_B\;,
\end{equation*}
concluding the proof.
\end{proof}

%% file: fk.tex
In~\cite{dunjko2016blind} the authors show that a measurement-based protocol for blind delegation of quantum circuits, the Universal Blind Quantum Computing protocol of Broadbent et al.~\cite{broadbent2009universal} (BFK) can be constructed from the ideal functionality \RSPB\ and classical communication channels. Their result builds upon the work of Dunjko et al.~\cite{dunjko2014composable}, who showed composable security of the BFK protocol in the AC framework. Dunjko et al.\ also showed composable security of a blind and verifiable variant of the BFK protocol introduced by Fitzsimons and Kashefi~\cite{fitzsimons2017unconditionally} (FK). 
Both protocols are designed to delegate a computation that is expressed in the model of measurement-based quantum computing (MBQC). In this model, a quantum computation is implemented by preparing a graph state (a collection of qubits that are entangled according to the structure of a graph) and then performing adaptive measurements on the qubits in the graph state.
The main difference between the FK protocol and the BFK protocol is the use of traps to ensure verifiability in FK. Informally, trap qubits are qubits initialized in a $\ket{+_\theta}$ state, with $\theta$ chosen uniformly at random from $\Theta$, and such that all neighbors of the trap qubit in the underlying graph state are initialized in a random computational basis state; these are called dummy qubits. The role of the dummy qubits is to isolate the trap qubits from the computation, so that the prover's measurements on the trap qubits can be verified independently of the computation. This isolation happens because, in the specific implementation of MBQC used by the FK and BFK protocols, the graph state is prepared by entangling $\ket{+_\theta}$ states using the Controlled-$Z$ operation. But note that this operation does not create entanglement if either of its input qubits is a state in the computational basis.
We sketch the structure of the FK protocol in Figure~\ref{fig:fk}, at a level that is sufficient to follow the arguments in this section; we refer to the description of Protocol 7 and Protocol 8 in~\cite{fitzsimons2017unconditionally} for full details. 

\begin{figure}[htbp]
\rule[1ex]{16.5cm}{0.5pt}\\
The inputs to the protocol are an error parameter $\eps>0$, a unitary quantum circuit $C$, and a classical input $\psi_A$. 
\begin{enumerate}
\item Alice selects a set of measurement angles  $\{\phi_i\}_{i\in\{1,\ldots,N\}}$, such that each $\phi_i \in \Theta$, that implement the computation specified by $C$. (For clarity we omit the choice of graph and flow.) Alice also selects a set of dummy qubit locations $D\subseteq \{1,\ldots,N\}$ and trap locations $T\subseteq \{1,\ldots,N\}$. Alice determines an update function $C(i,\phi_i,\theta_i,r_i,s)$. 
\item Alice selects angles  $\{\theta_i\}_{i\in\{1,\ldots,N\}}$  uniformly at random from $\Theta$, $\{r_i\}_{i\in\{1,\ldots,N\}}$ uniformly at random from $\{0,1\}$ and $\{d_i\}_{i\in D}$ uniformly at random from $\{0,1\}$. She initializes values $\{s_i\}_{i\in\{1,\ldots,N\}}$ to $0$. 
\item Alice prepares qubits in the state $\ket{d_i}$ for $i\in D$, and $Z^{d'_i} \ket{+_{\theta_i}}$ for $i\notin D$, where $d'_i$ is a predetermined function of $\{d_j\}_{j\in\{1,\ldots,N\}}$, and sends the qubits one by one to Bob.
\item For $i$ from $1$ to $N$:
\begin{enumerate}
\item Alice computes an angle $\delta_i = C(i,\phi_i,\theta_i,r_i,s)$ and sends it to Bob. 
\item Bob returns $b_i\in\{0,1\}$.
\item Alice sets $s_i\leftarrow b_i+r_i$. 
\end{enumerate}
\item Alice accepts if $s_i=r_i$ for all $i\in T$. She returns the state contained in the output qubits of the computation. 
\end{enumerate}
\rule[1ex]{16.5cm}{0.5pt}
\caption{Summary of the FK protocol. For an explanation of the notation and more details, see Protocol 8 in~\cite{fitzsimons2017unconditionally}.}
\label{fig:fk}
\end{figure}

Our goal in this section is to show that by replacing the quantum communication channel used by Alice to send single qubits to Bob in the FK protocol with BRSP we obtain a protocol that implements the ideal $\mS_{verif}^{blind}$ resource for blind and verifiable delegated computation. Importantly, the resulting protocol involves only classical communication and is composable.
We proceed in a number of incremental steps.

Note first that the set of single-qubit states prepared by the verifier in the FK protocol, $\ket{+_\theta}$ for $\theta\in\Theta$ and $\ket{0},\ket{1}$, is precisely the set of states that can be generated using the BRSP resource. 
We therefore define two variants of the FK protocol which we call \RSPV-FK and RSP-FK. The former is identical to the FK protocol, except Alice  uses the ideal resource \RSPV\ in order to prepare the states she is supposed to send to Bob in FK. RSP-FK is the same, except Alice uses the BRSP protocol to perform this preparation.

From the description of the FK protocol given in Figure~\ref{fig:fk}, it is clear that the number of times Alice uses the ideal \RSPV\ functionality or the BRSP protocol respectively is equal to the number of qubits she sends to the prover. As was shown in~\cite{kashefi2017optimised}, there exist graph states such that this number is linear in the size of the quantum circuit she wishes to delegate.
Given this, we can show the following:

\begin{lemma} \label{lemma:rspfk.rspvfk}
Let $\delta_{BRSP} > 0$ be the error tolerance parameter of the BRSP protocol used by Alice in RSP-FK, $\delta_{FK} > 0$ be the error (soundness) of FK, $T > 0$ the size of the computation Alice wishes to delegate to Bob\footnote{Alternatively, we can say that $T$ is the size of Alice's input and assume that she always delegates a universal circuit.} and $\lambda > 0$ the security parameter used in BRSP. Then, RSP-FK constructs \RSPV-FK, computationally, within distance $2\delta_{FK} + O(T \eps)$, where $\eps = O(\delta_{BRSP}^{c}) + \negl(\lambda)$, for some constant $c > 0$.
\end{lemma}

\begin{proof}
From Theorem~\ref{thm:brsp.rspv} we know that if $\delta_{BRSP}$ is the error of BRSP, then BRSP implements \RSPV\ to within computational distance $\eps$. We also know that \RSPV-FK involves $O(T \, \log(1/\delta_{FK}))$ uses of the \RSPV\ functionality, since to achieve error $\delta_{FK}$ the FK protocol uses $O(T \, \log(1/\delta_{FK}))$ qubits \cite{kashefi2017optimised}). In RSP-FK these are replaced with calls to BRSP. The compositionality theorem of AC (see~\cite{maurer2011abstract, dunjko2014composable}) implies that each replacement comes at an additive cost of $\eps$. In other words, up to an error $O(T \; \log(1/\delta_{FK}) \eps)$, RSP-FK behaves exactly the same as \RSPV-FK. The fact that \RSPV-FK has error $\delta_{FK}$ means that (conditioned on acceptance) it arrives at the correct result, except with error $\delta_{FK}$. The same will be true of RSP-FK, with the added error of $O(T \, \log(1/\delta_{FK}) \eps)$ stemming from the use of BRSP.
A triangle inequality leads us to conclude that RSP-FK implements \RSPV-FK within distance $2\delta_{FK} + O(T \, \log(1/\delta_{FK}) \eps)$.
\end{proof}
As a point of clarification, $\delta_{FK}$ represents the maximum deviations from the correct outcomes of the respective protocols, conditioned on Alice accepting. Also note that Alice can make the $O(T \; \log(1/\delta_{FK}) \eps)$ term be of order $\delta_{FK}$ by taking $\delta_{BRSP} = (\delta_{FK}/(T \, \log(1/\delta_{FK})))^{1/c}$.


\noindent We now show the following:
\begin{lemma} \label{lemma:rspvfk.fk}
Let $\delta_{FK} > 0$ be the error of FK. \RSPV-FK implements FK within distance $2\delta_{FK}$.
\end{lemma}
\begin{proof}
We denote the two protocols as $\pi^{\RSPV-FK} = ( \pi_A^{\RSPV-FK}, \pi_B^{\RSPV-FK})$ and $\pi^{FK} = ( \pi_A^{FK}, \pi_B^{FK})$ respectively. Additionally, let $\mathcal{R}^{cq}$ denote a resource consisting of classical and quantum channels. We will show that
\begin{equation}
\pi_A^{\RSPV-FK} \mathcal{R}^{cq} \pi_B^{\RSPV-FK} = \pi_A^{FK} \mathcal{R}^{cq} \pi_B^{FK} \;,\quad \quad
\pi_A^{\RSPV-FK} \mathcal{R}^{cq} \approx_{2\delta} \pi_A^{FK} \mathcal{R}^{cq} \sigma_{B}\;,
\end{equation}
for some simulator $\sigma_B$.

Correctness is immediate: if Alice and Bob behave honestly in both \RSPV-FK and in FK the results are statistically indistinguishable.

For security note the following. In \RSPV-FK, for each use of \RSPV\  Alice chooses the preparation bases for the $\ket{+_{\theta}}$ states at random, and for the dummies she consistently chooses the $Z$ basis. The only difference between this and the actual FK protocol is that because she is using the \RSPV\ functionality, Bob can force Alice to abort in the preparation stage by triggering the $ERR$ flag.
To show that \RSPV-FK implements FK, we need to show that the simulator, $\sigma_B$, interacting with $\pi_A^{FK} \mathcal{R}^{cq}$, can make its interaction with Bob indistinguishable from that of $\pi_A^{\RSPV-FK} \mathcal{R}^{cq}$. First of all, to match \RSPV-FK in terms of inputs and outputs it must be that for each qubit to be prepared the simulator receives the $c$ bit from Bob indicating whether he wants to cause the current preparation to abort, as per the specification of \RSPV\ (see Figure~\ref{fig:rspv}).

Consider a simulator that works as follows. The simulator first collects all the qubits sent by Alice through $\mathcal{R}^{cq}$. Then, for each qubit that it is supposed to send to Bob, it first receives the bit $c$ corresponding to that qubit. If $c=0$, indicating to not abort, the simulator sends that qubit to Bob. Otherwise, it sends the $ERR$ flag to Bob and also causes Alice to abort\footnote{The simulator can cause Alice to abort by providing random responses to the measurement outcomes she expects from Bob. The probability that these responses will match all of Alice's expected outcomes on the trap states is exponentially small. In other words, Alice will abort with probability $1 - \exp(-O(T))$. This will mean that \RSPV-FK implements FK within distance $2\delta_{FK} + \exp(-O(T))$ but since we will always consider $\delta_{FK} = \Omega(\exp(-T))$, we omit this inverse exponential term.}. 
If the simulator sends all of the qubits to Bob (i.e. there was no abort in the preparation stage) it then acts as a classical channel between Alice and Bob, forwarding the messages Alice sends (step 4.(a)) to Bob and then forwarding his responses to Alice (step 4.(b)).

From the soundness of FK it follows that at the end of either $\pi_A^{\RSPV-FK} \mathcal{R}^{cq}$ or $\pi_A^{FK} \mathcal{R}^{cq} \sigma_{B}$ the state of the system (conditioned on acceptance) is $\delta_{FK}$-close to 
the correct output. Applying the triangle inequality, the output states in the two situations are $2\delta_{FK}$-close to each other. This concludes the proof.\qedhere
\end{proof}
Note that the result of Lemma~\ref{lemma:rspvfk.fk} holds within statistical distance. Of course, the result is also true when restricting to the computationally efficient setting, since the simulator is efficient (it simply needs to store states and forward messages received from Alice and Bob).

Finally, we use the following result from~\cite{dunjko2013composable}:
\begin{lemma}[Lemma C.1 in~\cite{dunjko2013composable}]
\label{lem:FK}
If the FK protocol is run with parameters such that it has
error $\delta_{FK}$, then it is
$4\sqrt{2}\delta_{FK}^{1/4}N^2$-blind-verifiable, where $N$ is the
dimension of the subsystem of Alice's input which is quantum.
\end{lemma}

For convenience we restrict our attention to classical inputs for Alice. In this case it follows that FK (with classical input) with soundness parameter $\delta_{FK}$ implements the ideal $\mS_{verif}^{blind}$ resource within distance $O(\delta_{FK}^{1/4})$. With this fact, we can show the main result of this section.

\begin{theorem} \label{thm:dqc}
The RSP-FK protocol with error $\delta_{FK} > 0$ and security parameter $\lambda > 0$ implements the ideal $\mS_{verif}^{blind}$ resource within distance $O(\delta_{FK}^{1/4}) + \negl(\lambda)$.
\end{theorem}
\begin{proof}
From Lemmas~\ref{lemma:rspfk.rspvfk} and~\ref{lemma:rspvfk.fk} we see that RSP-FK with error $\delta_{FK}$ and security parameter $\lambda > 0$ implements FK with error $O(\delta_{FK}) + \negl(\lambda)$. Combining this with Lemma~\ref{lem:FK} leads us to conclude that RSP-FK implements the $\mS_{verif}^{blind}$ resource within distance $O(\delta_{FK}^{1/4}) + \negl(\lambda)$.
\end{proof}

The result of Theorem~\ref{thm:dqc} states that using BRSP together with the FK protocol yields a protocol that is computationally indistinguishable from the ideal blind-verifiability functionality. If we take the distinguishing advantage to be $\delta > 0$, what will be the total complexity (in terms of total of number of operations performed by the verifier) of RSP-FK for a computation of size $T$? From Theorem~\ref{thm:dqc} it follows that in order to implement $\mS_{verif}^{blind}$ to within distance $\delta$ we need to perform RSP-FK with soundness error $\delta^4$. Implementing the FK protocol so that it achieves soundness error $\delta^4$ requires $O(T \, \log(1/\delta^4))$ operations, where $T$ is the size of the computation. The change from $\delta$ to $\delta^4$ only increases the overhead by a constant factor, so that overall the prover requires $O(T \, \log(1/\delta))$ operations to implement FK. In our case, however, for each state sent by the verifier to the prover, the verifier executes BRSP. Thus the overhead is $O( C_{BRSP} \, T \, \log(1 / \delta) )$, where $C_{BRSP}$ is the cost of running one instance of BRSP. If we wish to achieve soundness error $\delta^4$ in RSP-FK, BRSP needs to have error at most $(\delta^4/(T \,\log(1/\delta)))^{1/c}$. The specific constant $c$ can be determined from the proof of Theorem~\ref{thm:rigidity} to be $c = 1/3$. Thus we can estimate the cost of BRSP as $C_{BRSP} = (T^{3}/\delta^{12}) \log^3(1/\delta) \poly(\lambda)$, where $\lambda$ is the security parameter. This gives a total cost of $O( (T^{4}/\delta^{12}) \, \log^4(1 / \delta) \, \poly(\lambda) )$.
Note that this is the cost of implementing the ideal blind-verifiable resource. If we merely wish to implement FK itself, the cost would be $O( (T^{4}/\delta^{3}) \, \log^4(1 / \delta) \, \poly(\lambda) )$, since we would not incur the $1/\delta \rightarrow 1/\delta^4$ increase stemming from Lemma~\ref{lem:FK}.

%% file: ntcf.tex
\section{Claw-free functions with adaptive hardcore}
\label{sec:entcf}

Our construction relies on a variant of a cryptographic primitive called a ``noisy trapdoor claw-free family (NTCF),'' introduced in~\cite{brakerski2018cryptographic}, and its extension to an ``extended noisy trapdoor claw-free family (ENTCF),'' given in~\cite{mahadev2018classical}. We rely on definitions and notation from~\cite[Section 3]{brakerski2018certifiable} and~\cite[Section 4]{mahadev2018classical_arxiv}. 

A key property of an NTCF is the adaptive hardcore bit property, property 4. in~\cite[Definition 3.1]{brakerski2018certifiable}. We need a slightly stronger variant of the property, that works over $\Z_8$ instead of $\Z_2$. The property we need is formulated in the following definition. 

\begin{definition}\label{def:z8-hc}
Let $\lambda$ be a security parameter. Let $\sX$ and $\sY$ be finite sets.
 Let $\mathcal{K}_{\mathcal{F}}$ be a finite set of keys. A NTCF family 
$$\mathcal{F} \,=\, \big\{f_{k,b} : \sX\rightarrow \mathcal{D}_{\sY} \big\}_{k\in \mathcal{K}_{\mathcal{F}},b\in\{0,1\}}$$
is said to have \emph{adaptive $\Z_8$ hardcore} if it satisfies the following conditions, for some integer $w$ that is a polynomially bounded function of $\lambda$. 
\begin{enumerate}
\item For all $b\in \{0,1\}$ and $x\in \sX$, there exists a set $\dset_{k,b,x}\subseteq \Z_8$ such that $\Pr_{d\leftarrow_U \Z_8^w}[d\notin \dset_{k,b,x}]$ is negligible, and moreover there exists an efficient algorithm that checks for membership in $\dset_{k,b,x}$ given $k,b,x$ and the trapdoor $t_k$. 
\item There is an efficiently computable injection $\inj:\sX\to \Z_8^w$, such that $\inj$ can be inverted efficiently on its range, and such that the following holds. For any $y\in\mY$, define functions $\hat{\theta}:\Z_8^w\to\{0,1,2,3\}$ and $\hat{v}:\Z_8^w\to\{0,1\}$ as the unique values such that $d\cdot (J(x_0)+J(x_1)) \bmod 8 = \hat{\theta}(d) + 4\hat{v}(d)$, where for $b\in\{0,1\}$, $x_b=\Inv_\mF(t_k,b,y)$, if  $d\in \dset_{k,0,x_0}\cap \dset_{k,1,x_1}$,\footnote{The sets $G_{k,b,x}$ are defined in~\eqref{eq:def-d}.} and $ \hat{\theta}(d)= \hat{v}(d)=\bot$ otherwise.
Then if
\begin{eqnarray*}\label{eq:defsetsH}
H_k &=& \big\{(b,x_b,d,\theta,v)\,|\; b\in \{0,1\},\; (x_0,x_1)\in \mathcal{R}_k,\;  (\theta,v)=(\hat{\theta}(d),\hat{v}(d))\big\}\;,\text{\footnotemark}\\
\overline{H}_k &=& \{(b,x_b,d,\theta,v)\,|\; (b,x,d,\theta,v\oplus 1) \in H_k\big\}\;,
\end{eqnarray*}
\footnotetext{Note that although both $x_0$ and $x_1$ are referred to to define the set $H_k$, only one of them, $x_b$, is explicitly specified in any $4$-tuple that lies in $H_k$.}
then for any quantum polynomial-time procedure $\mathcal{A}$ there exists a negligible function $\mu(\cdot)$ such that 
\begin{equation}\label{eq:adaptive-hardcore}
\Big|\Pr_{(k,t_k)\leftarrow \textrm{GEN}_{\mathcal{F}}(1^{\lambda})}[\mathcal{A}(k) \in H_k] - \Pr_{(k,t_k)\leftarrow \textrm{GEN}_{\mathcal{F}}(1^{\lambda})}[\mathcal{A}(k) \in\overline{H}_k]\Big| \,\leq\, \mu(\lambda)\;.
\end{equation}
Similarly, if for $w\in\{0,1,2,3\}$,
\begin{eqnarray*}\label{eq:defsetsH2}
H^{(w)}_k &=& \big\{(b,x_b,d,\theta)\,|\; b\in \{0,1\},\; (x_0,x_1)\in \mathcal{R}_k,\;  \theta=\hat{\theta}(d)+w\big\}\;,
\end{eqnarray*}
then for any quantum polynomial-time procedure $\mathcal{A}'$ there exists a negligible function $\mu(\cdot)$ such that for all $w\in\{1,2,3\}$,
\begin{equation}\label{eq:adaptive-hardcore-2}
\Big|\Pr_{(k,t_k)\leftarrow \textrm{GEN}_{\mathcal{F}}(1^{\lambda})}[\mathcal{A}'(k) \in H^{(0)}_k] - \Pr_{(k,t_k)\leftarrow \textrm{GEN}_{\mathcal{F}}(1^{\lambda})}[\mathcal{A}'(k) \in {H}^{(w)}_k]\Big| \,\leq\, \mu(\lambda)\;.
\end{equation}
\end{enumerate}
\end{definition}

\subsection{The adaptive hardcore property}
\label{sec:hardcore}

It is straightforward to verify that the same construction of an NTCF introduced in~\cite{brakerski2018certifiable} has adaptive $\Z_8$ hardcore. 

\begin{lemma}\label{lem:z8-hc}
The NTCF family introduced in~\cite{brakerski2018certifiable} has adaptive $\Z_8$ hardcore, i.e.\ it satisfies item 2. in Definition~\ref{def:z8-hc}.
\end{lemma}

The proof of Lemma~\ref{lem:z8-hc} is very similar to the adaptive hardcore bit condition shown in~\cite[Lemma 4.7]{brakerski2018certifiable}, with some modifications to obtain a hardness statement $\bmod 8$ instead of $\bmod 2$. We indicate the main changes needed, referring directly to statements from~\cite[Section 4.4]{brakerski2018certifiable}. 

The main step of the proof consists in showing the following Lemma, a direct analogue of~\cite[Lemma 4.2]{brakerski2018certifiable}. The main difference is the requirement on $\hat{d}$. For a string $x\in\{0,1\}^n$ we write $|x|_H$ for the Hamming weight of $x$.

\begin{lemma}\label{lem:hardcore-1}
Let $q$ be a prime, $\ell,n\geq 1$ integers, and $\*C\in \mZ_q^{\ell\times n}$ a uniformly random matrix. With probability at least $1-q^\ell\cdot 2^{-\frac{n}{8}}$ over the choice of $\*C$ the following holds. For a fixed $\*C$, all $\*v\in\mZ_q^\ell$ and $\hat{d}\in \{0,1\}^n$ such that $|\hat{d}|_H \geq \frac{n}{8}$, the distribution of $(\hat{d}\cdot s \bmod 8)$, where $s$ is uniform in $\{0,1\}^n$ conditioned on $\*C\*s = \*v$, is within  statistical distance $O(q^{\frac{3\ell}{2}} \cdot 2^{-\frac{n}{80}})$ of the uniform distribution over $\{0,1\}$. 
\end{lemma}

The first change in the proof of~\cite[Lemma 4.2]{brakerski2018certifiable} required to obtain Lemma~\ref{lem:hardcore-1} is to the definition of a moderate matrix: 

\begin{definition}\label{def:moderate}
Let $\*b\in \mZ_q^n$. We say that $\*b$ is \textnormal{moderate} if it contains at least $\frac{n}{4}$ entries whose unique representative in $(-q/2,q/2]$ has its absolute value in the range $(\frac{q}{32},\frac{3q}{32}]$. A matrix $\*C\in \mZ_q^{\ell\times n}$ is moderate if its entire row span (except $0^n$) is moderate.
\end{definition}

\begin{lemma}\label{lem:moderate}
Let $q$ be prime and $\ell,n$ be integers. Then 
$$\Pr_{\*C\leftarrow_U \mZ_q^{\ell\times n}}\big(\text{$\*C$ is moderate}\big)\,\geq \, 1 - q^\ell \cdot 2^{-\frac{n}{32}}\;.$$ 
\end{lemma}

\begin{proof}
The proof is identical to~\cite[Lemma 4.5]{brakerski2018certifiable}, except for replacing $q/8$ with $q/32$. 
\end{proof}

\begin{lemma}\label{lem:singled}
Let $\*C\in \mZ_q^{\ell\times n}$ be an arbitrary moderate matrix and let $\hat{d}\in\{0,1\}^n$ be such that $|\hat{d}|\geq\frac{n}{8}$. Let $s$ be uniform over $\{0,1\}^n$ and consider the random variables $\*v = \*C\*s \bmod q$ and $z = \hat{d}\cdot s \bmod 2$. Then $(\*v,z)$ is within total variation distance at most $2q^{\frac{\ell}{2}}\cdot 2^{-\frac{n}{80}}$ of the uniform distribution over $\mZ_q^{\ell}\times\{0,1\}$. 
\end{lemma}

The proof of the lemma is similar to the proof of~\cite[Lemma 4.5]{brakerski2018certifiable}, with a small difference due to the $\bmod 8$ condition. This is where the additional requirement that $|\hat{d}|\geq\frac{n}{10}$ (as opposed to simply $\hat{d}\neq 0$ in~\cite{brakerski2018certifiable}) is used. 

\begin{proof}
Let $f$ be the probability density function of $(\*v,z)$. Interpreting $z$ as an element of $\mZ_8$, let $\hat{f}$ be the Fourier transform over $\mZ_q^{\ell}\times \mZ_8$. Let $U$ denote the density of the uniform distribution over $\mZ_q^{\ell}\times \mZ_8$. Applying the Cauchy-Schwarz inequality,
\begin{align}
\frac{1}{2}\big\| f - U \big \|_1 &\leq 2\sqrt{{q^\ell}} \big\| {f} - {U} \big\|_2 \notag \\
&= \frac{1}{2} \big\| \hat{f} - \hat{U} \big\|_2 \notag \\
&= \frac{1}{2} \Big( \sum_{(\hat{\*v},\hat{z})\in \mZ_q^\ell \times \mZ_8\backslash\{(\*0,0)\}} \big| \hat{f}(\hat{\*v},\hat{z})\big|^2\Big)^{1/2}\;,\label{eq:fc-1}
\end{align}
where the second line follows from Parseval's identity, and for the third line we used $\hat{f}(\*0,0)=\hat{U}(0,0)=1$ and $\hat{U}(\hat{\*v},\hat{z})=0$ for all $(\hat{\*v},\hat{z})\neq(0^\ell,0)$. To bound~\eqref{eq:fc-1} we estimate the Fourier coefficients of $f$. Denoting $\omega_{8q} = e^{-\frac{2\pi i}{8q}}$, for any $(\hat{\*v},\hat{z}) \in \mZ_q^\ell \times \mZ_2$ we can write 
\begin{align}
\hat{f}({\hat{\*v}},{\hat{z}}) &= \Es{\*s}\Big[\omega_{8q}^{(2\cdot {\hat{\*v}}^TC + q\cdot {\hat{z}}\hat{\*d}^T)\*s}\Big]\notag\\
&= \Es{\*s}\big[\omega_{8q}^{\*w^T\*s}\big] \notag\\
&= \prod_i\Es{s_i}\big[\omega_{8q}^{w_is_i}\big]\;,\label{eq:fc-3}
\end{align}
where we wrote $\*w^T = 8\cdot {\hat{\*v}}^T\*C + q\cdot {\hat{z}}\hat{\*d}^T\in \mZ_{8q}^n$. 

We first bound $\hat{f}(0^\ell,\hat{z})$ for $\hat{z}\in \mZ_8\backslash\{0\}$. In this case~\eqref{eq:fc-3} simplifies to 
\begin{align}
\big|\hat{f}({\hat{\*v}},{\hat{z}}) \big| &= \prod_{i:\hat{d}_i=1}\big|\Es{s_i}\big[e^{-\frac{2i\pi\hat{z}}{8}s_i }\big]\big|\notag\\
&= \prod_{i:\hat{d}_i=1}\Big|\cos\Big(\frac{\pi}{2} \frac{\hat{z}}{4}\Big)\Big|\notag\\
&\leq \prod_{i:\hat{d}_i=1} \cos\Big(\frac{\pi}{8}\Big) \,\leq\, 2^{-\frac{n}{80}}\;.\label{eq:fc-1a}
\end{align}

Next we observe that for all $i\in\{1,\ldots,n\}$ such that the representative of $({\hat{\*v}}^T\*C)_i$ in $(-q/2,q/2]$ has its absolute value in $(\frac{q}{32},\frac{3q}{32}]$ it holds that $\frac{w_i}{q}\in (\frac{1}{4},\frac{3}{4}]\bmod 1$, in which case
\begin{equation}
\big|\Es{s_i}[\omega_{8q}^{w_is_i}]\big| \,=\, \Big|\cos\Big(\frac{\pi}{2}\cdot \frac{w_i}{q}\Big)\Big| \,\leq\, \cos\Big(\frac{\pi}{8}\Big)\,\leq\, 2^{-\frac{1}{10}}\;.
\end{equation}
Since $\*C$ is moderate, there are at least $\frac{n}{4}$ such entries, so that from~\eqref{eq:fc-3} it follows that $|\hat{f}({\hat{\*v}},{\hat{z}})|\leq 2^{-\frac{n}{40}}$ for all $\hat{\*v} \neq \*0$. Recalling~\eqref{eq:fc-1} and~\eqref{eq:fc-1a}, the lemma is proved.  
\end{proof}

The proof of Lemma \ref{lem:hardcore-1} follows from Lemma \ref{lem:singled} exactly as in~\cite{brakerski2018certifiable}, and we omit the details. With Lemma~\ref{lem:hardcore-1} in hand, the proof of the adaptive $\Z_8$ condition, item 2. in Definition~\ref{def:z8-hc}, is very similar to the proof of the adaptive hardcore bit condition in~\cite{brakerski2018certifiable}. The main change needed is in the definition of the sets $G_{k,b,x}$, for $k = (\*A,\*A\*s + \*e), b\in\{0,1\}$ and $x\in\sX$, that is defined as follows:
\begin{equation}\label{eq:def-d}
\dset_{k,b,x}\,=\,\Big\{d\in \{0,1\}^w\,\Big|\; \big|I_{b,x}(d)_{\{b\,\frac{n}{2},\ldots,b\,\frac{n}{2}+\frac{n}{2}\}}\big|_H \geq \frac{n}{4} \Big\}\;,\end{equation}
where $I_{b,x}(d)$ is the vector whose each coordinate is obtained by taking the inner product mod $2$ of the corresponding block of $\lceil\log q\rceil$ coordinates of $d$ and of $\inj(x)\oplus \inj(x-(-1)^b\*1 )$, where  $\inj:\sX\to\{0,1\}^w$  is such that $\inj(x)$ returns the binary representation of $x\in\sX$ and $\*1 \in \mZ_q^n$ is the vector with all its coordinates equal to $1\in\mZ_q$.

\subsection{The collapsing property}

The following lemma shows that any ENTCF has the \emph{collapsing} property, introduced by Unruh~\cite{unruh2016computationally}.

\begin{lemma}\label{lem:collapse}
Let $(\mF,\mG)$ be an ENTCF family. Let $\phi = \sum_{y\in\mY} \proj{y} \otimes \phi_y$ be a state that can be prepared efficiently, given as input a key $k\in \mK_\mF\cup \mK_\mG$. Let $\Pi = \{\Pi^{(b,x_b)}\}$ be an efficiently implementable POVM such that $\Tr(\Pi^{(b,x_b)} \phi_y)=0$ if $f_{k,b}(x_b)\neq y$ (if $k\in\mK_\mF)$ or $g_{k,b}(x_b)\neq y$ (if $k\in\mK_\mG$). Then there is no efficient procedure such that, given $k\leftarrow \textrm{GEN}_{\mathcal{F}}(1^{\lambda})$  and $y$ distributed according to $\Tr(\phi_y)$, the procedure has a non-negligible advantage in distinguishing $\phi_y$ from $\phi'_y = \Pi^{(0,x_0)} \phi_y\Pi^{(0,x_0)}+\Pi^{(1,x_1)} \phi_y\Pi^{(1,x_1)}$, where for $b\in\{0,1\}$, $x_b$ is such that $f_{k,b}(x_b)=y$.  
\end{lemma}

\begin{proof}
Suppose for contradiction that there exists such a procedure. Since the procedure is efficient, using the property of injective invariance of an ENTCF (Definition 4.2 in~\cite{mahadev2018classical_arxiv}) it should produce computationally indistinguishable outcomes  given $k\leftarrow \textrm{GEN}_{\mathcal{F}}(1^{\lambda})$ or  $k\leftarrow\textrm{GEN}_{\mathcal{G}}(1^{\lambda})$. In the second case, $\phi'_y=\phi_y$, so that no such procedure exists. 
\end{proof}